\documentclass[12pt]{article}
\usepackage{amsmath}
\usepackage{graphicx}
\usepackage{psfrag}
\usepackage{epsf, epstopdf}
\usepackage{enumerate}
\usepackage[square,numbers]{natbib}
\usepackage{url} 
\usepackage{amsthm}
\usepackage{multirow}

\usepackage{tikz}
\usepackage{color}
\usepackage{amssymb}
\usepackage{latexsym}
\usepackage{xr}
\usepackage{epstopdf}
\usepackage{subcaption}

\def\RR{\mathbb R}
\def\ZZ{\mathbb Z}

\def\NN{\mathbb N}

\def\EE{\mathbb E}
\def\PP{\mathbb P}

\def\PP{\mathbb P}
\def\X{{\bf X}}

\def\bmu{\boldsymbol{\mu}}
\def\btheta{\boldsymbol{\theta}}
\def\bleta{\boldsymbol{\eta}}

\def\bphi{\boldsymbol{\phi}}

\def\bGamma{\boldsymbol{\Gamma}}
\def\bgamma{\boldsymbol{\gamma}}

\newtheorem{Proposition}{Proposition}[section]
\newtheorem{Lemma}{Lemma}[section]

\theoremstyle{definition}
\newtheorem{Remark}{Remark}[section]

\begin{document}

\def\spacingset#1{\renewcommand{\baselinestretch}%
{#1}\small\normalsize} \spacingset{1}

\title{Latent Gaussian Count Time Series}
  \author{
	Yisu Jia \\ University of North Florida \and Stefanos Kechagias \\ SAS Institute \and James Livsey \\ United States Census Bureau  \and Robert Lund \thanks{
		Robert Lund's research was partially supported by the grant \textit{NSF DMS 1407480}.} \\  University of California - Santa Cruz \and Vladas Pipiras \thanks{Vladas Pipiras's research was partially supported by the grant \textit{NSF DMS 1712966}.}\\ University of North Carolina - Chapel Hill}
	
\date{June 7, 2021}
\maketitle

\bigskip
\begin{abstract}
This paper develops the theory and methods for modeling a stationary count time series via Gaussian transformations.   The techniques use a latent Gaussian process and a distributional transformation to construct stationary series with very flexible correlation features that can have any pre-specified marginal distribution, including the classical Poisson, generalized Poisson, negative binomial, and binomial structures.  Gaussian pseudo-likelihood and implied Yule-Walker estimation paradigms, based on the autocovariance function of the count series,
are developed via a new Hermite expansion. Particle filtering and sequential Monte Carlo methods are used to conduct likelihood estimation.   Connections to state space models are made.  Our estimation approaches are evaluated in a simulation study and the methods are used to analyze a count series of weekly retail sales.
\end{abstract}

\vfill

\newpage
\spacingset{1.5} 

\section{Introduction}
\label{intro}

This paper develops the theory and methods for modeling a stationary discrete-valued time series by transforming a Gaussian process.  Since the majority of discrete-valued time series involve integer counts supported on some subset of $\{ 0, 1, \ldots \}$, we isolate on this support set.  Our methods are based on a copula-style transformation of a latent Gaussian stationary series and are able to produce any desired count marginal distribution.  It is shown that the proposed model class produces the  most flexible pairwise correlation structures possible, including negatively dependent series.  Model parameters are estimated via 1) a Gaussian pseudo-likelihood approach, developed from some new Hermite expansion techniques, which use only the mean and the autocovariance of the series, 2) an implied Yule-Walker moment estimation approch when the latent Gaussian process is an autoregression, and 3) a particle filtering (PF) / sequential Monte Carlo (SMC) approach that uses a state space model (SSM) representation of the transformation to approximate the true likelihood. Extensions to non-stationary settings, particularly those with covariates, are discussed.

The theory of stationary Gaussian time series is by now well developed.  A central result is that a stationary Gaussian series $\{ X_t \}_{t \in \ZZ}$ having the lag-$h$ autocovariance $\gamma_X(h) =\mbox{Cov}(X_t, X_{t+h})$ exist if and only if $\gamma_X$ is symmetric about lag zero and non-negative definite (see Theorem 1.5.1 in \cite{brockwell:davis:1991}). However, such a result does not hold for stationary count series having a certain prescribed marginal distribution (e.g, Poisson). In principle, distributional existence issues are checked with Kolmogorov's consistency criterion (see Theorem 1.2.1 in \cite{brockwell:davis:1991}); in practice, one needs a specified joint distribution to check for consistency.   Phrased another way, Kolmogorov's consistency criterion is not a constructive result and does not illuminate how to build stationary time series having a particular marginal distribution and correlation structure.  Perhaps owing to this, count time series have been constructed from a plethora of approaches over the years, as is next reviewed.

Drawing from the success of autoregressive moving-average (ARMA) models in describing stationary Gaussian series, early count authors constructed correlated count series from discrete ARMA (DARMA) and integer ARMA (INARMA) difference equation methods. Focusing on the first order autoregressive case for simplicity, a DAR(1) series $\{ X_t \}_{t=1}^T$ with specified marginal distribution $F_X(\cdot)$ is obtained by generating $X_1$ from $F_X(\cdot)$ and then at each subsequent time, either keeping the previous count value with probability $p$ or generating an independent copy of $F_X(\cdot)$ with probability $1-p$.  INAR(1) series are built via the thinned AR(1) equation $X_t = p \circ X_{t-1} + \epsilon_t$, where $\{ \epsilon_t \}$ is an IID count-valued random sequence and $\circ$ is a thinning operator defined by $p \circ Y =  {\rm B}(Y,p)$ for a binomial distribution ${\rm B}(n,p)$ with $n$ trials and success probability $p$. DARMA methods were initially explored in \cite{JacobsLewis_1978a}, but were subsequently discarded by practitioners because their sample paths often remained constant for long periods, especially in highly correlated cases; INARMA series are still used today. In contrast to their Gaussian ARMA brethren, DARMA and INARMA models, and their extensions in \cite{Joe_1996}, cannot produce negative autocorrelations.

The works \cite{Blight_1989} and \cite{CuiLund_2009} take a different approach, producing the desired count marginal distribution by combining IID copies of a correlated Bernoulli series $\{ B_t \}$ built from a stationary renewal sequence.  Explicit autocovariance functions when $\{ B_t \}$ is made by binning (clipping) a stationary Gaussian sequence into zero-one categories are derived  in \citep{Livsey_etal_2018}. While these models can have negative correlations, they do not necessarily produce the most negatively correlated count structures possible.  Also, some important count marginal distributions, including generalized Poisson, are not easily built from these methods.  The results here easily generate any desired count marginal distribution.  Other count model classes studied include Gaussian processes rounded to their nearest integer \citep{Kachour_Yao_2009}, hierarchical Bayesian count model approaches \citep{Gammerman_etal_2015},  and others (see \cite{fokianos:2012} and \cite{Davis_etal_2015} for recent reviews).  Each approach has some drawbacks.

The models here impose a fixed marginal distribution for the counts.  This is in contrast
to generalized ARMA methods (GLARMA), which typically posit conditional distributions
in lieu of marginal distributions, with model parameters typically being random.
As \cite{Asmussen_2014} shows in the Poisson case, once the randomness of the parameters
is taken into account, the true marginal distribution of the series can be far from the
posited conditional distribution. This said, the literature on GLARMA and other conditional
models is extensive \citep{GARMA, zheng:2015}.   See \cite{Dunsmuir} for a recent review of
GLARMA models.

A time series analyst generally needs four features in a count model: 1) general marginal
distributions; 2) the most general correlation structures possible, both positive and
negative; 3) the straight-forward accomodation of covariates; and 4) a well performing
and computationally feasible likelihood inference approach.  All previous count classes
fail to accommodate one or more of these tenets.  This paper's purpose is to introduce
and study a count model class that, for the first time, simultaneously achieves all four features.  Our model employs a latent Gaussian process and a copula-style transformation.  This type of construction has recently shown promise in spatial statistics \citep{DeOliveria_2013, HanDeOliveria_2016}, multivariate modeling \citep{Smith2012, song2013}, and regression \citep{masarotto:2012}, but the theory has yet to be developed for count series (\cite{masarotto:2012, Lennon_2016} provide some partial results). Our objectives here are several-fold. On a methodological level, it is shown, through some newly derived Hermite polynomial expansions, that accurate and efficient numerical quantification of the correlation structure of this count model class is feasible.  Based on a result in \cite{whitt:1976bivariate}, the class is shown to produce the most flexible pairwise correlation structures possible, positive or negative (see Remark \ref{r:max-min-corr} below).  Connections to both importance sampling schemes, where the popular GHK sampler in \cite{masarotto:2012} is adapted to our needs, and to the SSM and SMC literature, which allow natural extensions of the GHK sampler and likelihood evaluation, are made.  The methods are tested on both synthetic and real data.

The works \cite{masarotto:2012, Lennon_2016} are perhaps the closest papers to this study.  While the general latent Gaussian construct adopted is the same, our work differs in that explicit autocovariance relations are developed via Hermite expansions, flexibility and optimality issues  of the model class are addressed, Gaussian pseudo-likelihood and implied least-squares parameter estimation approaches are developed, and both the importance sampling and SSM connections are explored in detail.  Additional connections to \cite{masarotto:2012, Lennon_2016} and to the spatial count modeling papers \cite{HanDeOliveria_2016, HanDeOliveria_2018} are later made.

The rest of this paper proceeds as follows.  The next section and Appendix A introduce our Gaussian transformation count model and establish its basic mathematical and statistical properties.  Section \ref{s:inference} and Appendix B move to estimation, developing three techniques: a Gaussian pseudo-likelihood approach, implied Yule-Walker estimation, and PF/SMC methods. Section \ref{s:simulations} and Appendix C present simulation results.  Section \ref{s:applications} and Appendix D analyze soft drink sales counts at one location of the now defunct Dominick’s Finer Foods retail chain. This series exhibits overdispersion, negative lag one autocorrelation, and dependence on a price reduction (sales) covariate, which illustrates the flexibility of our approach.   Section \ref{s:conclusions} concludes with comments and suggestions for future research.

\section{Theory}
\label{methods}

We seek to construct a strictly stationary time series $\{ X_t \}$ having marginal distributions from any family of count distributions supported in $\{ 0, 1, \ldots \}$, including the Binomial, Poisson, mixture Poisson, negative binomial, generalized Poisson, and Conway-Maxwell-Poisson distributions. The later three distributions are over-dispersed (their variances are larger than their respective means), which is the case for many observed count time series.

Let $\{ X_t \}_{t \in \ZZ}$ be the stationary count time series of interest. Suppose that one wants the marginal cumulative distribution function (CDF) of $X_t$  for each $t$ of interest to be $F_{X}(x) = \PP [X_t \leq x]$, depending on a vector $\boldsymbol{\theta}$ containing all CDF model parameters.  The series $\{ X_t \}$ will be modeled through
\begin{equation}
\label{e:Y-GX}
X_t = G(Z_t),~~~\text{ where }~~~G(z) = F_{X}^{-1}(\Phi(z)), \quad z \in \RR,
\end{equation}
and $\Phi(\cdot)$ is the CDF of a standard normal variable and $F_{X}^{-1}(u) = \inf \{t: F_{X}(t) \geq u \}$, $u \in (0,1)$, is the generalized inverse (quantile function) of the CDF $F_{X}$. The process $\{ Z_t \}_{t \in \ZZ}$ is standard Gaussian for each fixed $t$, but possibly correlated in time:
\begin{equation}
\label{e:X-mean-var}
\EE[ Z_t ] = 0, \quad \EE [ Z_t^2] =1, \quad \rho_Z(h)=:\mbox{Corr}(Z_t,Z_{t+h})=\EE[ Z_t Z_{t+h}].
\end{equation}
This approach has been used in \cite{Smith2012, masarotto:2012, HanDeOliveria_2016, Lennon_2016} with good results.  The autocovariance function (ACVF) of $\{ Z_t \}$, denoted by $\gamma_Z(\cdot)$, is the same as the autocorrelation function (ACF) due to standard normality and depends on another vector $\boldsymbol{\eta}$ of ACVF parameters.

As expanded on in Section \ref{s:particles}, (\ref{e:Y-GX}) can be viewed as a SSM:
\begin{eqnarray*}
	\mbox{\rm State equation} & : &  p(z_t|z_{t-1},\ldots, z_1)\ \mbox{\rm governing latent Gaussian dynamics}; \\
	\mbox{\rm Observation equation} & : &  \PP(X_t=k|z_t) = 1_{A_k}(z_t)\ \mbox{\rm with the set $A_k$ defined below.}
\end{eqnarray*}
Here, $p(\cdot|\cdot)$ is notation for an arbitrary conditional distribution.

This model has alternative names in other literature.  For example, \cite{HChen_2001} call this setup the normal to anything (NORTA) procedure in operations research, whereas \cite{Grigoriu_2007} calls this a translational model in mechanical engineering.  Our goal is to give a reasonably complete analysis of the probabilistic and statistical properties of these models.

The construction in (\ref{e:Y-GX}) ensures that the marginal CDF of $X_t$ is indeed $F_X(\cdot)$.  Elaborating, the probability integral transformation theorem shows that $\Phi(Z_t)$ has a uniform distribution over $(0,1)$ for each $t$; a second application of the result justifies that $X_t$ has the marginal distribution $F_X(\cdot)$ for each $t$. Moreover, temporal dependence in $\{ Z_t \}$ will induce temporal dependence in $\{ X_t \}$ as quantified below. For notation, let $\gamma_{X}(h) = \EE [ X_{t+h} X_t ]- \EE [X_{t+h}]  \EE [X_t]$ denote the ACVF of $\{ X_t \}$.

\subsection{Relationship between autocovariances}
\label{s:model-rel-acvfs}

The autocovariance functions of $\{ X_t \}$ and $\{ Z_t \}$ can be related using Hermite expansions (see Chapter 5 of \cite{pipiras:2017}).
In particular, using the Hermite polynomials $H_k(z) = (-1)^k e^{z^2/2} \frac{d^k}{dz^k} (e^{-z^2/2})$, $z\in \RR$
we can expand the $L^2$ function $G$ as
\begin{equation}
	\label{e:G-herm-exp}
	G(z) =\EE [G(Z_0)]+ \sum_{k=1}^\infty g_{k} H_k(z)
\end{equation}
where the \textit{Hermite coefficients} $g_k$ are given by
\begin{equation}
\label{e:herm-coef}
g_{k} = \frac{1}{k!} \int_{- \infty}^\infty G(z) H_k(z) \frac{e^{-z^2/2}dz}{\sqrt{2\pi}} = \frac{1}{k!} \EE [ G(Z_0) H_k(Z_0)],
\end{equation}
for a standard normal variable $Z_0$.
The relationship between $\gamma_X(\cdot)$ and $\gamma_Z(\cdot)$ is key and is extracted from Chapter 5 of \cite{pipiras:2017}:
\begin{equation}
\label{e:Y-X-cov}
\gamma_{X}(h) = \sum_{k=1}^\infty k! g_{k}^2 \gamma_{Z}(h)^k =: g(\gamma_{Z}(h)),
\end{equation}
where $g(u) = \sum_{k=1}^\infty k! g_{k}^2  u^k$.  For $h=0$,  \eqref{e:Y-X-cov} yields Var$(X_t) = \gamma_{X}(0) = \sum_{k=1}^\infty k! g_{k}^2$, which depends only on the marginal parameters in $\boldsymbol{\theta}$.  Moreover, the ACF of $\{ X_t \}$ is
\begin{equation}
\label{e:Y-X-cor}
\rho_{X}(h) = \sum_{k=1}^\infty \frac{k! g_k^2}{\gamma_{X}(0)} \gamma_{Z}(h)^k =: L(\rho_{Z}(h)),
\end{equation}
where
\begin{equation}
\label{e:Y-X-cor-func}
L(u) = \sum_{k=1}^\infty \frac{k! g_{k}^2}{\gamma_{X}(0)}  u^k =: \sum_{k=1}^\infty \ell_{k} u^k,
\end{equation}
and $\ell_k= k! g_k^2/\gamma_X(0)$.  The function $L$  maps $[-1,1]$ into
(but not necessarily onto) $[-1,1]$.  For future reference, note that $L(0)=0$
and $L(1) = \sum_{k=1}^\infty \ell_{k} = 1$.  Using  (\ref{e:G-herm-exp}) and
$\EE[ H_k(Z_0) H_\ell(-Z_0)]=(-1)^k k! 1_{[ k = \ell]}$ gives
$L(-1) = \mbox{\rm Corr}(G(Z_0),G(-Z_0))$; however, $L(-1)$ is not necessarily $-1$
in general. As such, $L(\cdot)$ ``starts" at $(-1,L(-1))$, passes through $(0,0)$,
and connects to $(1,1)$.  Examples are given in Figure 2 of
Appendix A.

We call the quantity $L(\cdot)$ a {\it link} function, and the coefficients $\ell_{k}$, $k\geq 1$,
{\it link coefficients}. (Sometimes, slightly abusing terminology, we also use these terms
for $g(\cdot)$ and $g_k^2 k!$, respectively.)  A key feature in (\ref{e:Y-X-cov}) is that
the effects of the marginal CDF $F_{X}(\cdot)$ and the ACVF $\gamma_{Z}(\cdot)$ are
``decoupled'' in the sense that the correlation parameters in $\{ Z_t \}$ do not
influence the $g_k$ coefficients in (\ref{e:Y-X-cov}) --- this is useful later in estimation.

Further properties and the numerical calculation of the link function and the Hermite coefficients
are discussed in Appendix A. The computation of the Hermite
coefficients, in particular, is feasible due to the following lemma, which is proved in Appendix
A.

\begin{Lemma} \label{r:coeff-conv}
	If $\EE[ X_t^p ] < \infty$ for some $p>1$, then the coefficients
	$g_k$ satisfy
	\begin{equation}
		g_{k} =  \frac{1}{k!\sqrt{2\pi}} \sum_{n=0}^\infty e^{-\Phi^{-1}(C_{n})^2/2} H_{k-1}(\Phi^{-1}(C_{n})),
	\label{e:herm-coef-expr2}
	\end{equation}
where $C_{n} = \PP[ X_t \leq n]$. (When $\Phi^{-1}(C_n) = \pm \infty$ (that is, $C_n=0$ or $1$), the summand $e^{-\Phi^{-1}(C_{n})^2/2} H_{k-1}(\Phi^{-1}(C_{n}))$ is interpreted as zero.)
\end{Lemma}

Returning to the relationship between $\rho_X(h)$ and $\rho_Z(h)$, from (\ref{e:Y-X-cor}), one can see that
\begin{equation}
|\rho_X(h)| \leq |\rho_Z(h)|,
\label{e:corrless}
\end{equation}
which implies that a positive $\rho_Z(h)$ leads to a positive $\rho_X(h)$. A negative $\rho_Z(h)$ produces a
negative $\rho_X(h)$ since $L(u)$ is, in fact, monotone increasing (see Proposition
A.1 in Appendix  A) and crosses zero at $u=0$ (the negativeness
of $\rho_X(h)$ when $\rho_Z(h)<0$ can also be deduced from the nondecreasing nature of
$G$ via an inequality on page 20 of \cite{Tong} for Gaussian variables).

\begin{Remark}
\label{r:srd-lrd}
The short- and long-range dependence properties of $\{ X_t \}$ can be extracted from those of $\{ Z_t \}$.  Recall that a time series $\{Z_t\}$ is short-range dependent (SRD) if $\sum_{h=-\infty}^\infty |\rho_Z(h)|<\infty$. According to one definition, a series $\{Z_t\}$ is long-range dependent (LRD) if $\rho_Z(h) = Q(h) h^{2d-1}$, where $d \in (0,1/2)$ is the LRD parameter and $Q$ is a slowly varying function at infinity \citep{pipiras:2017}. The ACVF of such LRD series satisfies $\sum_{h=-\infty}^\infty |\rho_Z(h)|=\infty$. If $\{Z_t\}$ is SRD, then so is $\{X_t\}$ by (\ref{e:corrless}).  On the other hand, if $\{Z_t\}$ is LRD with parameter $d$, then $\{X_t\}$ can be either LRD or SRD. The conclusion depends, in part, on the Hermite rank of $G(\cdot)$, which is defined as $r = \min \{ k \geq 1: g_k \neq 0 \}$. Specifically, if $d \in (0,(r-1)/2r)$, then $\{X_t\}$ is SRD; if $d \in ((r-1)/2r,1/2)$, then $\{ X_t \}$ is LRD with parameter $r(d-1/2)+1/2$ (see \cite{pipiras:2017}, Proposition 5.2.4).
\end{Remark}

The model in (\ref{e:Y-GX}) admits the following structure:   if $Z_t$ and $Z_s$ are independent, then so are $X_t$ and $X_s$.   It follows that if $\{ Z_t \}$ is stationary and $q$-dependent, than both $\{ Z_t \}$ and $\{ X_t \}$ must be $q$th order moving-average time series.   Unfortunately, no analogous autoregressive structure holds; in fact, if $\{ Z_t \}$ is a first order autoregression, then $\{ X_t \}$ may not be an autoregression of any order (this can be inferred from \cite{Kedem_1980}).

\begin{Remark}
\label{r:max-min-corr}
The construction in (\ref{e:Y-GX}) yields models with the most flexible correlations possible for $\mbox{Corr}(X_{t_1},X_{t_2})$ for two variables $X_{t_1}$ and $X_{t_2}$ with the same marginal distribution $F_X$. Indeed, let $\rho_- = \min \{ \mbox{Corr}(X_{t_1},X_{t_2}): X_{t_1},X_{t_2} \sim F_X\}$ and define $\rho_+$ similarly with $\min$ replaced by $\max$. Then, as shown in
Theorem 2.5 of \cite{whitt:1976bivariate},
\[
\rho_+ = \mbox{Corr}(F_X^{-1}(U), F_X^{-1}(U)) = 1, \quad \rho_- =  \mbox{Corr}(F_X^{-1}(U), F_X^{-1}(1-U)),
\]
where $U$ is a uniform random variable over $(0,1)$. Since $U \stackrel {\cal D} {=} \Phi(Z)$ and $1-U \stackrel {\cal D} {=} \Phi(-Z)$ for a standard normal random variable $Z$, the maximum and minimum correlations $\rho_+$ and $\rho_-$ are indeed achieved with (\ref{e:Y-GX}) when $Z_{t_1}=Z_{t_2}$ and  $Z_{t_1}=-Z_{t_2}$, respectively. The preceding statements are non-trivial for $\rho_-$ only since $\rho_+=1$ is attained whenever $X_{t_1}=X_{t_2}$. It is worthwhile to compare this to the discussion following (\ref{e:Y-X-cor-func}). Finally, all correlations in $(\rho_-,\rho_+)=(\rho_-,1)$ are achievable since $L(u)$ in (\ref{e:Y-X-cor-func}) is continuous in $u$. The flexibility of correlations for Gaussian copula models in the spatial context was also noted and studied in \cite{HanDeOliveria_2016}, especially in comparison to a class of hierarchical, e.g.\ Poisson, models.
\end{Remark}

The preceding remark settles autocovariance flexibility issues for stationary count series.  Flexibility is a concern when the series is negatively correlated, an issue arising, for example, with hurricane counts in \cite{Livsey_etal_2018} and chemical process counts in \cite{Kachour_Yao_2009}.  Since any general count marginal distribution can also be achieved, the model class is quite general.

\subsection{Covariates}
\label{s:cov-1}
There are situations where stationarity is not desired.  Such scenarios can often be accommodated by simple variants of the above setup.  For concreteness, consider a situation where a vector $\textbf{M}_t$ of $J$ non-random covariates is available to
explain the series at time $t$. If one wants $X_t$ to have the marginal distribution $F_{\boldsymbol{\theta}(t)}(\cdot)$, where $\boldsymbol{\theta}(t)$ is a vector-valued function of $t$ containing marginal distribution parameters, then simply set
\begin{equation}
X_t = F^{-1}_{\boldsymbol{\theta}(t)}(\Phi(Z_t))
\label{e:timevarying}
\end{equation}
and reason as before.   We do not recommend modifying $\{ Z_t \}$ for the covariates as this may bring process existence issues into play.

Generalized linear models link functions (not to be confused with $L(\cdot)$ in (\ref{e:Y-X-cor})--(\ref{e:Y-X-cor-func})) can be used when parametric support set bounds are encountered.  For example, a Poisson regression with correlated errors can be formulated via a parameter vector $\boldsymbol{\beta}$ of regression coefficients with $\boldsymbol{\theta}(t)= \EE[ X_t ] = \mbox{exp}(\boldsymbol{\beta}^\prime \boldsymbol{M}_t)$.  Here, the exponential link guarantees that the Poisson parameter is positive.  The above construct requires the covariates to be non-random; should covariates be random, marginal distributions may change from $F_{\boldsymbol{\theta}(t)}$.

\subsection{Particle filtering and state space model connections}
\label{s:particles}

This subsection studies the implications of the latent structure of our model, especially as it relates to SSMs and importance sampling approaches.  This will be used to construct PF/SMC approximations of various quantities, and in goodness-of-fit assessments. Our main reference is \cite{doucet:2001}. As in that monograph, let  $z_{0:t} = \{Z_0=z_0, \ldots, Z_t=z_t\}$, $x_{0:t} = \{X_0=x_0, \ldots, X_t=x_t\}$, and $p(\cdot)$ and $p(\cdot | \cdot)$ denote joint and conditional probabilities (or their densities, depending on the context).  For example, $p(z_{0:t}|x_{0:t})$ denotes the conditional density of $Z_{0:t}$ given  $x_{0:t}$. Similarly, let $\mathbb{E}[\cdot | x_{0:t}]$ denote conditional expectation given $x_{0:t}.$ The SSM formulation starts by specifying $p(z_{t+1}|z_{0:t})$ and $p(x_t|z_t)$. While $\{ Z_t \}$ is often first order Markov, implying that  $p(z_{t+1}|z_{0:t}) = p(z_{t+1}|z_{t})$, this is not necessary.

To specify $p(z_{t+1}|z_{0:t})$ in our stationary Gaussian case, we compute the best one-step-ahead linear prediction of $Z_{t+1}$ from $z_{0:t}$ given by $\hat{Z}_{t+1}=\phi_{t0}Z_t+\ldots+\phi_{tt}Z_0$. The coefficients $\phi_{ts}$, $s \in \{0, \ldots, t\}$, can be computed recursively in $t$ from the ACF of $\{ Z_t \}$ via the classical  Durbin-Levinson (DL) or the Innovations algorithm, for examples.  As a convention, we take $\hat{Z}_0 = 0$. Let $r_t^2 = \EE[ (Z_t - \hat Z_t)^2]$ be the corresponding unconditional mean squared prediction error. With this notation,
\begin{equation}
\label{e:smc-model}
	p(z_{t+1}|z_{0:t}) \stackrel{{\cal D}}{=}  \mathcal{N}(\hat{z}_{t+1},r^2_{t+1}),
\end{equation}
where  $\hat{z}_{t+1}=\phi_{t0}z_t+\ldots+\phi_{tt}z_0$. Again, $\{ Z_t \}$ does not have to be Markovian (of any order). On the other hand, with \eqref{e:Y-GX},
\begin{equation}
\label{e:smc-model-mass}
	p(x_t|z_t) = \delta_{G(z_t)}(x_t) =
	\left\{
		\begin{array}{cl}
			1, & \mbox{if}\ x_t = G(z_t), \\
			0, & \mbox{otherwise},
		\end{array}
	\right.
\end{equation}
where $\delta_y(x)$ is a unit point mass at $\{ y \}$. The equations in \eqref{e:smc-model} and \eqref{e:smc-model-mass} constitute the SSM representation of \eqref{e:Y-GX}.

In inference and related tasks for SSMs, the basic goal is to compute the conditional expectation $\mathbb{E}[v(Z_{0:t}) | x_{0:t}]$ for some function $v$. This is often carried out through an importance sampling algorithm such as sequential importance sampling (SIS),
which generates $N$ independent particle trajectories $Z^i_{0:t}, i \in \{ 1, \ldots,  N \}$, from a  proposal distribution $\pi(z_{0:t}|x_{0:t})$ and approximates the conditional expectation as
\begin{equation}
\label{e:smc-approx}
\mathbb{E}[v(Z_{0:t}) | x_{0:t}] \approx \displaystyle\sum_{i=1}^N v(Z^i_{0:t}) \widetilde{w}^i_t=: \hat{\mathbb{E}}[v(Z_{0:t})|x_{0:t}],
\end{equation}
where
\begin{equation}
\label{e:smc-weights}
	\widetilde{w}^i_t = \frac{w(Z^i_{0:t})}{\sum_{i=1}^N w(Z^i_{0:t})}, \quad w(z_{0:t}) = \frac{p(z_{0:t}|x_{0:t})}{\pi(z_{0:t}|x_{0:t})},
\end{equation}
are the (normalized) importance weights (see \cite{doucet:2001} and \cite{liu:1998}).  Furthermore, in SIS,
\begin{equation}
\label{e:smc-weights-2}
	\widetilde{w}^i_t \propto \widetilde{w}^i_{t-1}w_t(Z^i_{0:t}), \quad w_t(z_{0:t}) = \frac{p(x_t|z_t)p(z_t|z_{0:t-1})}{\pi(z_t|z_{0:t-1}, x_{0:t})}
\end{equation}
(see (1.6) in \cite{doucet:2001}, which is adapted to a possibly non-Markov setting by replacing $p(z_t|z_{t-1})$ with $p(z_t|z_{0:t-1})$). The two probability terms in the numerator of $w_{t}(z_{0:t})$ in \eqref{e:smc-weights-2}
constitute the SSM, whereas the denominator relates to the proposal distribution.

We suggest the following proposal distribution and the resulting SIS algorithm for our model. Take
\begin{equation}
\label{e:proposal}
	\pi(z_t|z_{0:t-1},x_{0:t}) \stackrel{{\cal D}}{=}  \mathcal{N}_{A_{x_t}}(\hat{z}_t, r_t^2),
\end{equation}
where $\mathcal{N}_{A}$ denotes a normal distribution restricted to the set $A$, and
\begin{equation}
\label{e:truncation set}
A_k = \{z: \Phi^{-1}(C_{k-1}) \leq z \le \Phi^{-1}(C_{k}) \}.
\end{equation}
The role of $A_k$ stems from the fact
\begin{equation}
\label{e:truncation set-2}
	k = G(z) \Leftrightarrow z \in A_k
\end{equation}
(i.e., the count value $k$ is obtained if and only if $Z_t \in A_k$; see the expression (A.2) for $G(z)$). In particular, for $\{ Z_t^i \}$ generated from the proposal distribution \eqref{e:proposal}, the term $p(x_t|Z_t^i)$ in the incremental weight $w_t(Z^i_{0:t})$ of \eqref{e:smc-weights-2} is always set to unity.   The rest of the incremental weights are calculated as
\[
	w_t(z_{0:t}) = \frac{p(z_t|z_{0:t-1})}{\pi(z_t|z_{0:t-1},x_{0:t})}
	 			 = \frac{e^{-\frac{(z_t-\hat{z}_t)^2}{2r_t^2}}/(2\pi r^2_t)^{1/2}}
	  				{e^{-\frac{(z_t-\hat{z}_t)^2}{2r_t^2}} / [(2\pi r^2_t)^{1/2} \times \mathbb{P} (N(\hat{z}_t,r_t^2) \in A_{x_t})] }
\]
\begin{equation}	
\label{e:explicit-truncNormal}	
	= \mathbb{P} (\mathcal{N}(\hat{z}_t,r_t^2) \in A_{x_t})  =\Phi\Big(\frac{\Phi^{-1}(C_{x_t})-\hat{z}_t}{r_t}\Big) - \Phi\Big(\frac{\Phi^{-1}(C_{x_t-1})-\hat{z}_t}{r_t}\Big)=: w_t(\hat{z}_t).
\end{equation}
The choice of the proposal distribution is largely motivated by $\PP(X_t=k|Z_t^i)=1_{A_k}(Z_t^i)$ and the explicit form in \eqref{e:explicit-truncNormal} for the incremental weights $w_t(z_{0:t})$. Optimality considerations are mentioned in Remark B.3.

The following steps summarize our SIS algorithm.

\noindent {\bf Sequential Importance Sampling (SIS):} For $i \in \{1, \ldots, N \}$, where $N$ represents the number of particles, initialize the weight $w^i_0=1$ and the latent series $Z_0^i$ by
\begin{equation}
\label{e:X0-new}
Z_0^i \stackrel{{\cal D}}{=} \mathcal{N}_{A_{x_0}}(0,1).
\end{equation}
Then, recursively over $t=1, \ldots, T$, perform the following steps:
\begin{itemize}
	\item [1:] Compute $\hat{Z}_t^i$ with the DL or other algorithm using the previously generated values of $Z_0^i, \ldots, Z_{t-1}^i$.
	\item [2:] Update the series $Z_t^i$ and the importance weight $w_t^i$ via
\begin{equation}
\label{e:X-update-pred-again}
	Z_t^i \stackrel{{\cal D}}{=} \mathcal{N}_{A_{x_t}}(\hat{Z}_t^i,r_t^2), \quad w^i_t = w^i_{t-1} w_t(\hat{Z}_t^i),
\end{equation}
where	$w_t(z)$ is defined in \eqref{e:explicit-truncNormal}.
\end{itemize}

\begin{Remark} \label{r:sis-idea}
For $i \in \{ 1, \ldots, N \}$, the constructed path $\{ Z^i_t \}_{t=0}^T$ is one of the $N$ independent ``particles'' used to approximate the conditional expectation in (\ref{e:smc-approx}). Equation \eqref{e:X-update-pred-again} ensures that for each $i$, the path $\{ Z^i_t \}_{t=0}^T$ obeys the restriction $G(Z^i_t)=x_t$ and matches the temporal structure of $\{ Z_t \}$.  These two properties show that $\{ Z^i_t \}_{t=0}^T$ is a realization of the latent Gaussian stationary series producing $X_t = x_t$ for all $t$. Finally, we note where the model parameters enter into the SIS algorithm. The marginal distribution parameters $\boldsymbol{\theta}$ enter through the form of $C_x$ in (\ref{e:explicit-truncNormal}), whereas the temporal dependence parameters $\boldsymbol{\eta}$ enter through the one-step-ahead prediction coefficients $\phi_{ts}, s \in \{0,\ldots,t\}$, in the calculation of $\hat{Z}^i_t$ in Step 1 of the algorithm, and through the prediction error $r_t$.
\end{Remark}

To compute the model likelihood, several known formulas applicable in the (general) SIS setting are needed.   The relation
\[
\frac{p(z_{0:t}|x_{0:t})}{\pi(z_{0:t}|x_{0:t})} = \displaystyle\prod_{s=0}^t w_s(z_{0:s}) \frac{p(x_0)}{p(x_{0:t})},
\]
produces
\[
\mathbb{E}[w_t^i v(Z_{0:t}^i)|x_{0:t}]p(x_0) = \mathbb{E}[v(Z_{0:t})|x_{0:t}]p(x_{0:t}).
\]
In particular (with $v(\cdot) \equiv1 $),
\begin{equation}
\label{e:weight-expectation}
	\mathbb{E}[w_t^i]p(x_0) = p(x_{0:t}).
\end{equation}
To conduct prediction, we use Equation (1.2) in \cite{doucet:2001} to get
\begin{equation}
\label{e:cond expectation}
	\mathbb{E}[v(X_{t+1})|x_{0:t}] = \mathbb{E}\Big[   \mathbb{E}[v( G(\hat{Z}_{t+1}))|Z_{0:t}]|x_{0:t}\Big] =: \mathbb{E}[D_{v,t+1}(\hat{Z}_{t+1}) | x_{0:t}],
\end{equation}
where
\begin{equation}
\label{e:cond expectation-2}
	D_{v,t+1}(z) = \mathbb{E}\big[v(G(\mathcal{N}(z,r^2_{t+1})))\big] = \displaystyle\int_{\mathbb{R}}
				 v(G(z_{t+1}))\frac{1}{\sqrt{2\pi r^2_{t+1}}}e^{-\frac{(z_{t+1}-z)^2}{2r_{t+1}^2}}dz_{t+1},
\end{equation}
since $\hat{Z}_{t+1}|z_{0:t} \stackrel{{\cal D}}{=} \mathcal{N}(\hat{z}_{t+1},r_{t+1}^2)$.
In view of \eqref{e:cond expectation} and \eqref{e:smc-approx}, the following prediction approximation arises:
\begin{equation}
\label{e:prediction-approx}
	\EE[v({X}_{t+1})|x_{0:t}] \approx \sum_{i=1}^N \frac{w^i_t}{\Omega_{N,t}} D_{v,t+1}(\hat{Z}_{t+1}^i) =: \hat{\mathbb{E}}[v(X_{t+1})|x_{0:t}],
\quad
	\Omega_{N,t}=\sum_{i=1}^N w_t^i.
\end{equation}

Appendix B further connects our model and algorithm to the popular GHK sampler, hidden Markov models (HMMs), and PF and SMC techniques.

The SIS algorithm has a fundamental weakness called ``weight degeneracy": as the algorithm propagates through an increasing number of iterations, a large number of the normalized weights become negligible.  As a result. only a few particles ``contribute" in the likelihood approximation.
Following the developments in the SMC (see \cite{doucet:2001}, \cite{liu:2008} and \cite{chopin:2020}) and HMM literatures (Sections 10.4.1 and 10.4.2 in \cite{douc:2014}), we modify the SIS algorithm by adding a resampling step (all future simulations and computations use resampling).
\vspace{.12in} \noindent {\bf Sequential Importance Sampling with Resampling (SISR):} Proceed as in the SIS algorithm, but modify Step 2 and add a resampling Step 3 as follows:
\begin{itemize}
\item [2:] Modify Step 2 of the SIS by setting
\begin{equation}
\label{e:sisr-update}
\widetilde Z_t^i \stackrel {\cal D} {=}  \mathcal{N}_{A_{x_t}}(\hat{Z}^i_{t},r_{t}^2)  , \quad \widetilde w^i_t = w^i_{t-1} w_t(\hat{Z}_t^i),
\quad
\widetilde \Omega_{N,t} = \sum_{i=1}^N \widetilde w^i_t.
\end{equation}
\item [3:] For each particle $i \in \{ 1, \ldots, N \}$, draw, conditionally and independently given $\{ (Z_s^i,w_s^i), s \leq t-1  \}$ and $\{ \widetilde Z_t^i \}$, a multinomial trial $I_t^i \ in \{ 0,1 \}$ for each $t$ and $i$ with the success probabilities
$\{ \widetilde w^i_t / \widetilde{\Omega}_{N,t} \}$ and set $Z_t^i = \widetilde Z_t^{I_t^i}$ and $w_t^i=1$.

\end{itemize}
While the resampling step removes particles with low weights, mitigating degeneracy
issues, it introduces additional estimator variance. We follow standard practice
and resample only when the variance of the weights exceeds a certain threshold,
quantified by the so-called \textit{effective sample size} defined as
$\mbox{ESS}(w^i_t) = ( \sum_{i=1}^{N} ( w^i_t/\Omega_{N,t} )^2 )^{-1}$, and the
resampling step is executed when $\mbox{ESS}(w^i_t)<N/2$ as in \cite{doucet:2009}. See also
Section 2.5.3 in \cite{liu:2008} for a justification of the ESS based on the Delta method.

\section{Inference}
\label{s:inference}
The model in \eqref{e:Y-GX} contains the parameters $\btheta$ in the marginal count distribution $F_X$ and $\bleta$ in the dependence structure of $\{ Z_t \}$. This section addresses inference questions, including parameter estimation and goodness-of-fit assessment.   Three methods are presented for parameter estimation:  Gaussian pseudo-likelihood, implied Yule-Walker moment methods, and full likelihood.  Gaussian pseudo-likelihood estimators, a time series staple, pretend that the series is Gaussian and maximize its Gaussian-based likelihood.   These estimators only involve the mean and covariance structure of the series, are easy to compute, and will provide a comparative basis for likelihood estimators.   They can also be used as initial guesses in gradient step-and-search likelihood optimizations.  Implied Yule-Walker techniques are moment based estimators applicable to the commonly encountered case where $\{ Z_t \}$ is a causal autoregression.   Likelihood estimators, the statistical gold standard and the generally preferred estimation technique, are based on the PF and SMC methods of the last section. Finally, we will not delve into a detailed statistical inference for the aforementioned methods:  while consistency and asymptotic normality are expected in some of the examined cases (e.g.\ likelihood estimation with an autoregressive $\{ Z_t \}$), a rigorous theoretical treatment is beyond the scope of this paper.

\subsection{Gaussian pseudo-likelihood estimation}
\label{s:inference-lse}
As in Section \ref{s:particles}, we work with observations $x_t$ for the times $t \in \{ 0, \ldots , T \}$ and set $\X = ( x_0, \ldots, x_T)'$.   Denote the likelihood of the model in \eqref{e:Y-GX} by
\begin{equation}
\label{e:full_lik}
{\cal L}_T(\boldsymbol{\theta}, \boldsymbol{\eta}) = \mathbb{P}(X_0 = x_0, X_1 = x_1, \ldots, X_T = x_T) = p(x_{0:T}).
\end{equation}
While this likelihood is a multivariate normal probability, it is difficult to calculate or approximate when $T$ is large.  For most count model classes, true likelihood estimation is difficult to conduct as joint distributions are generally intractable \citep{Davis_etal_2015}.
While Section \ref{s:inference-mle} below devises a well performing PF/SMC likelihood approximation (see also \cite{Smith2012}), we first consider a simple Gaussian pseudo-likelihood (GL) approach.  In a pseudo GL approach, parameters are estimated via
\begin{equation}
\label{e:GL}
(\hat{\btheta},\hat{\bleta} ) = \underset{\btheta,\bleta}{\mbox{argmax}}
\frac{e^{-\frac{1}{2}(\X-\bmu_{\btheta})'\Gamma_T(\btheta,\bleta)^{-1}(\X-\bmu_{\btheta})}}{(2\pi)^{(T+1)/2}|\Gamma_T(\btheta,\bleta)|^{1/2}},
\end{equation}
\noindent where $\boldsymbol{\mu}_{\boldsymbol{\theta}}=(\mu_{\boldsymbol{\theta}}, \ldots, \mu_{\boldsymbol{\theta}})^\prime$ is a $(T+1)$-dimensional constant mean vector. These estimators maximize the series' likelihood assuming the data are Gaussian, each component having mean $\mu_{\btheta}$, and all components having covariance matrix $\Gamma_T(\btheta,\bleta) = (\gamma_{X}(i-j))_{i,j=0}^T$. Time series analysts have been maximizing Gaussian pseudo likelihoods for decades, regardless of the series' marginal distribution, with often satisfactory performance.  The next section and Appendix C present a case where this approach works reasonably well, and one where it does not.  For large $T$, the pseudo GL approach is equivalent to least squares estimation, where the sum of squares $\sum_{t=0}^T (X_t - \EE[X_t| X_0, \ldots, X_{t-1}])^2$ is minimized (see Chapter 8 in \cite{brockwell:davis:1991}). The covariance structure of $\{ X_t \}$ was efficiently computed in Section 2; the mean $\mu_{\btheta}$ is usually explicitly obtained from the marginal distribution $F_X$ posited.  Numerical optimization of \eqref{e:GL} yields a Hessian matrix that can be inverted to obtain standard errors for the model parameters.   These standard errors can be asymptotically corrected for distributional misspecification via the sandwich methods of \cite{Freedman_2006}.

\subsection{Implied Yule-Walker estimation for latent AR models}
\label{s:inference-ils}

Suppose that $\{ Z_t \}$ follows the causal AR$(p)$ model $Z_t = \phi_1 Z_{t-1} + \ldots + \phi_p Z_{t-p} + \varepsilon_t$, where $\{ \varepsilon_t \}$ consists of  IID ${\cal N}(0,\sigma^2_\varepsilon)$ variables.  Here, $\sigma^2_\varepsilon$ depends on the autoregressive coefficients $\phi_1, \ldots, \phi_p$ in a way that induces $\EE[Z_t^2]=1$. The Yule-Walker equations are
\begin{equation}
\label{e:ARp-YW}
	\bphi = \bGamma_p^{-1} \bgamma_p,
\end{equation}
where $\bGamma_p = (\gamma_Z(i-j))_{i,j=1}^p$, $\bgamma_p = (\gamma_Z(1), \ldots, \gamma_Z(p))'$, and $\bphi = (\phi_1, \ldots,\phi_p)'$.   From (\ref{e:Y-X-cor}), note that
\begin{equation}
\label{e:ARp-cov-inv}
	\gamma_Z(h) = L^{-1}(\rho_X(h)),
\end{equation}
the inverse being justified via the strictly increasing nature of $L(u)$ in $u$.

Equations (\ref{e:ARp-YW}) and (\ref{e:ARp-cov-inv}) suggest the following estimation procedure. First, estimate the CDF parameter $\btheta$ directly from the counts; standard methods (e.g.\ method of moments) are typically available for this task. The estimated parameter $\hat{\btheta}$ defines an estimated link $\hat{ L}(u)$ through its estimated power series coefficients. From a numerical power series reversion procedure, one can now efficiently construct the inverse estimator $\hat{L}^{-1}(\rho)$.

Next, in view of (\ref{e:ARp-cov-inv}) and  (\ref{e:ARp-YW}), set
\begin{equation}
\label{e:ARp-cov-inv-est-ARp-YW-est}
	\hat{\gamma}_Z(h) = \hat{L}^{-1}(\hat{\rho}_X(h)),\quad \hat{\bphi} = \hat{\bGamma}_p^{-1} \hat{\bgamma}_p,
\end{equation}
where $\hat{\rho}_X(h)$ is the lag-$h$ sample autocorrelation of $\{ X_t \}$, and $\hat{\bGamma}_p$ and  $\hat{\bgamma}_p$ are defined analogously to the above using $\hat{\gamma}_Z(h)$ in place of $\gamma_Z(h)$.

\subsection{Particle filtering and sequential Monte Carlo likelihoods}
\label{s:inference-mle}
Using \eqref{e:cond expectation} and its notation, the true likelihood in
\eqref{e:full_lik} is
\begin{equation}\label{e:likelihood}
	{\cal L}_T(\btheta,\bleta) = p(x_0) \prod_{s=1}^{T} p(x_s | x_{0:s-1}) = p(x_0) \prod_{s=1}^{T}  \mathbb{E}[ 1_{\{x_s\}}(X_s)|x_{0:s-1}  ] = p(x_0)\prod_{s=1}^{T} \mathbb{E}[ w_s(\hat{Z}_s)|x_{0:s-1}],
\end{equation}
where \eqref{e:cond expectation} was used with $D_{1_{\{x_{s}\}},s}(z) = w_{s}(z)$ and $w_{s}(z)$ is defined and numerically computed akin to \eqref{e:explicit-truncNormal}.
The particle approximation of the likelihood is then
\begin{equation}
\label{e:likelihood-particle}
	\hat{{\cal L}}_T(\btheta,\bleta) =  p(x_0)\prod_{s=1}^{T} \hat{\EE}[ w_s(\hat{Z}_s)|x_{0:s-1} ];
\end{equation}
this uses the notation in \eqref{e:smc-approx} and supposes that the particles are generated by one of the methods in Section \ref{s:particles}.
The approximate PF maximum likelihood estimates satisfy
\begin{equation}
\label{e:mle-particles}
	(\hat{\btheta},\hat{\bleta}) = \mathop{\rm argmax}_{\btheta,\bleta} \hat{{\cal L}}_T(\btheta,\bleta).
\end{equation}

\begin{Remark}
\label{r:lik-sis}
With the SIS algorithm, (\ref{e:likelihood-particle}) reduces to
\begin{equation}
\label{e:likelihood-particle-sis}
\hat {\cal L}_T(\btheta,\bleta) =  p(x_0) \frac{1}{N} \sum_{i=1}^N w_T^i,
\end{equation}
which is consistent with \eqref{e:weight-expectation}.  The work \cite{masarotto:2012} also essentially implements (\ref{e:likelihood-particle-sis}). In contrast to \cite{masarotto:2012}, our approach includes a resampling step in the likelihood approximations, considers other estimation approaches (pseudo GL and implied Yule-Walker), and provides model diagnostic tools more specific to count series (the PIT histograms in Section 3.4 below).
\end{Remark}

To optimize the estimate $\hat{{\cal L}}_T(\btheta,\bleta)$, we employ a large number of particles (growing linearly with $T$)  and common random number (CRN) techniques, a standard practice that serves to smooth $\hat{{\cal L}}_T(\btheta,\bleta)$ somewhat by expressing its random quantities through parameter-dependent transformations of uniform random variables that remain constant for likelihood evaluations across distinct parameters. While the CRN procedure works well in SIS, it fails to ward against discontinuous $\hat{{\cal L}}_T(\btheta,\bleta)$ in our preferred SISR algorithm. An elegant solution to this issue for univariate state processes is proposed in  \cite{malik:2011}: first reorder the (real-valued) particles and then replace the discontinuous resampling CDF with a piecewise linear approximation.  More recent and well performing (but less straightforward) approaches such as the sequential quasi Monte Carlo and the $\mbox{SMC}^2$ algorithm are reviewed in detail in Chapters 13, 14, and 18 of \cite{chopin:2020} (see also the Chapter 19 references on controlled sequential Monte Carlo methods).   We do not pursue these issuees further here.

In our numerical implementations, gradient-free algorithms  from the R package optimx \cite{nash:2011} are used, which follows standard practices in optimizing noisy objective functions. These routines allow for boundary constraints and performed well in modest computing times for our sample sizes. On the other hand, we found less success with the more popular gradient-based quasi-Newton algorithm L-BFGS-B (gradients were computed via finite differences) as convergence instabilities and high-variance estimates were encountered.   However, promising recent developments for optimizing noisy objectives in \cite{berahas:2019} and \cite{shi:2021} were not explored.  A comprehensive investigation of these approaches and of the rich gradient-based SMC inference literature for our framework as in \cite{kantas:2015} is deferred to future work.

\subsection{Model diagnostics}
\label{s:inference-diagnostics}

The goodness-of-fit of count models is commonly assessed through probability integral transform (PIT) histograms and related tools \citep{czado:etal:2009, kolassa:2016}. These are based on the predictive distributions of $\{ X_t \}$, defined at time $t$ by
\begin{equation}
\label{e:predic-distr}
	P_{t}(y) = \PP( {X}_{t} \leq y | X_0=x_0,\ldots,X_{t-1}=x_{t-1}) = \PP(X_t\leq y|x_{0:t-1} ),\quad y \in \{0,1,\ldots\}.
\end{equation}
This quantity can be estimated through the PF/SMC methods in Section \ref{s:particles} as
\begin{equation}
\label{e:predic-distr-est}
	\hat{P}_{t}(y) = \sum_{\ell=0}^y \hat \EE [1_{\{ \ell \}} ( X_{t})|x_{0:t-1}] =
\sum_{\ell=0}^y \hat{\EE}[ D_{1_{\{ \ell \}},t} (\hat{Z}_{t})|x_{0:t-1}],
\end{equation}
which uses \eqref{e:cond expectation-2} and \eqref{e:prediction-approx} and supposes that the particles are generated by the SIS, SISR, or other algorithms. Similar to $D_{1_{\{x_{s}\}},s}(z) = w_{s}(z)$, note that $D_{1_{\{x\}},t}(z) = \widetilde{w}_{x,t}(z)$, where
\begin{equation}
\label{e:prediction-filtering-gf-special-3}
	\widetilde w_{x,t}(z) = \Phi\Big(\frac{\Phi^{-1}(C_{x}) - z}{r_{t}}\Big) - \Phi\Big(\frac{\Phi^{-1}(C_{x-1}) - z}{r_{t}}\Big)
\end{equation}
and $\widetilde w_{x_{t},t}(z) = w_t(z)$.

The (non-randomized) sample mean PIT is defined as
\begin{equation}
\label{e:mean-pit}
	\overline F(u) = \frac{1}{T+1} \sum_{t=0}^T F_t(u|x_t),\quad u\in [0,1],
\end{equation}
where
\begin{equation}
\label{e:mean-pit2}
	F_t(u|y) = \left\{
	\begin{array}{cl}
		0, & \mbox{if}\ u \leq P_t(y-1), \\
		\frac{u-P_t(y-1)}{P_t(y) - P_t(y-1)}, & \mbox{if}\  P_t(y-1) <u < P_t(y), \\
		1, & \mbox{if}\  u \geq P_t(y),
	\end{array}
	\right.
\end{equation}
which is estimated by replacing $P_t$ by $\hat{P}_t$ in practice. The PIT histogram with $H$ bins is defined as a histogram with the height $\overline F(h/H) - \overline F((h-1)/H)$ for bin $h\in \{1,\ldots,H\}$.

Another possibility considers model residuals based on
\begin{equation}
\label{e:res}
\hat{Z}_t = \mathbb{E}[Z_t|X_t=x_t] =
\frac{\exp(-\Phi^{-1}(C_{x_t-1})^2/2)-\exp(-\Phi^{-1}(C_{x_t})^2/2)}{\sqrt{2 \pi}(C_{x_t} - C_{x_t-1})},
\end{equation}
which is the estimated mean of the latent Gaussian process at time $t$ given $X_t$ only (not the entire past), where $\eqref{e:res}$ follows by direct calculations for the model \eqref{e:Y-GX} (using the estimated parameters $\btheta$ of the marginal distribution in the $C_k$s).  For a fitted underlying time series model with parameter $\bleta$, the residuals are then defined as the standard time series residuals $\hat{\epsilon}_t$ of this model fitted to the series $\hat{Z}_t$, after centering by the sample mean.

\subsection{Nonstationarity and covariates}

As discussed in Section \ref{s:cov-1}, covariates can be accommodated by allowing a time-varying parameter $\btheta$ in the marginal distribution.  With covariates, $\btheta$ at time $t$ is denoted by $\btheta(t)$. The GL and PF/SMC procedures are modified for $\btheta(t)$ as follows.

For the GL procedure, the covariance
$\mbox{Cov}(X_{t_1},X_{t_2}) = \mbox{Cov}(G_{\btheta(t_1)}(Z_{t_1}),G_{\btheta(t_2)}(Z_{t_2}))$
is needed, where $G$ is subscripted to signify dependence on $\btheta(t)$.  But as in \eqref{e:Y-X-cov},
\begin{equation}
\label{e:cov-cov}
\mbox{Cov}(X_{t_1},X_{t_2}) = \mbox{Cov}(G_{\btheta(t_1)}(Z_{t_1}),G_{\btheta(t_2)}(Z_{t_2})) =
\sum_{k=1}^\infty k! g_{\btheta(t_1),k} g_{\btheta(t_2),k} \gamma_{Z}(t_1-t_2)^k,
\end{equation}
where again, the subscript $\btheta(t)$ is added to the $g_k$s to indicate dependence on $t$. Numerically, evaluating \eqref{e:cov-cov} is akin to the task in \eqref{e:Y-X-cov}; in particular, both calculations are based on the Hermite coefficients $\{ g_k \}$.

For the PF/SMC approach, the modification is somewhat simpler: one just needs to replace $\btheta$ by $\btheta(t)$ at time $t$ when generating the underlying particles. For example, for the SIS algorithm, $\btheta(t)$ enters only through the $C_x$s in \eqref{e:explicit-truncNormal}, \eqref{e:X0-new}, and \eqref{e:X-update-pred-again}. This is because the covariates enter only through $\btheta$, the parameter controlling marginal distributions.

\section{A simulation study}
\label{s:simulations}
To evaluate our estimation methods, a simulation study considering several marginal distributions and dependence structures was conducted. Here, the classic Poisson count distribution $\mathcal{P}$ is examined (mixed Poisson and negative binomial simulations are presented in Appendix C), with $\{ Z_t \}$ taken from the ARMA$(p,q)$ class. All simulation cases are replicated 200 times for three distinct series lengths:  $T=100, 200$, and $400$.  For notation, estimates of a parameter $\zeta$ from Gaussian pseudo-likelihood (GL), implied Yule-Walker (IYW), and PF/SMC methods are denoted by $\hat{\zeta}_{GL}$, $\hat{\zeta}_{IYW}$, and $\hat{\zeta}_{PF}$, respectively.

We now consider the classical case where $X_t$ has a Poisson marginal distribution for each $t$ with mean $\lambda > 0$. To obtain $X_t$, the AR(1) process $Z_t = \phi Z_{t-1} + (1-\phi^2)^{1/2}\epsilon_t$, was simulated and transformed via \eqref{e:Y-GX} with $F=\mathcal{P}$;  $\EE[ Z_t^2] \equiv 1$ was induced by taking $\mbox{Var}(\epsilon_t) \equiv 1$. Twelve parameter schemes resulting from all combinations of $\lambda \in\{2, 5, 10 \}$ and $\phi \in \{ \pm 0.25, \pm 0.75 \}$ were considered.

Figure \ref{f:pois-ar1-sim-lam2} displays box plots of the parameter estimates when $\lambda=2$. In estimating $\lambda$, all methods perform reasonably well. When the lag-one correlation in $\{ Z_t \}$  (and hence also that in $\{ X_t \}$) is negative (right panel), $\hat{\lambda}_{GL}$, $\hat{\lambda}_{IYW}$, and $\hat{\lambda}_{PF}$ have smaller variability than the positively correlated case (left panel --- note the different y-axis scales on the panels). This is expected: the mean of $\{ X_t \}$ is $\lambda$, and the variability of the sample mean, one good estimator of the mean for a stationary series, is smaller for negatively correlated series than for positively correlated ones.  Note that $\hat{\phi}_{GL}$ is biased toward zero for both negatively and positively correlated
series, whereas $\hat{\phi}_{IYW}$ and $\hat{\phi}_{PF}$ only show bias when $\phi$ is
positive for the sample sizes $T=100$ and $T=200$. Overall, the PF/SMC estimates were
the least biased.  All estimates of $\phi$ have roughly similar variances. Simulations
with $\lambda=5$ and $\lambda=10$ produced analogous results with smaller values of
$\lambda$ yielding less variable estimates. This is again expected as the variance of
the Poisson distribution is also $\lambda$. Graphics of these box plots are omitted for
brevity's sake.

\begin{figure}[h!]
	\centerline{
	\includegraphics[width = 0.49\textwidth]{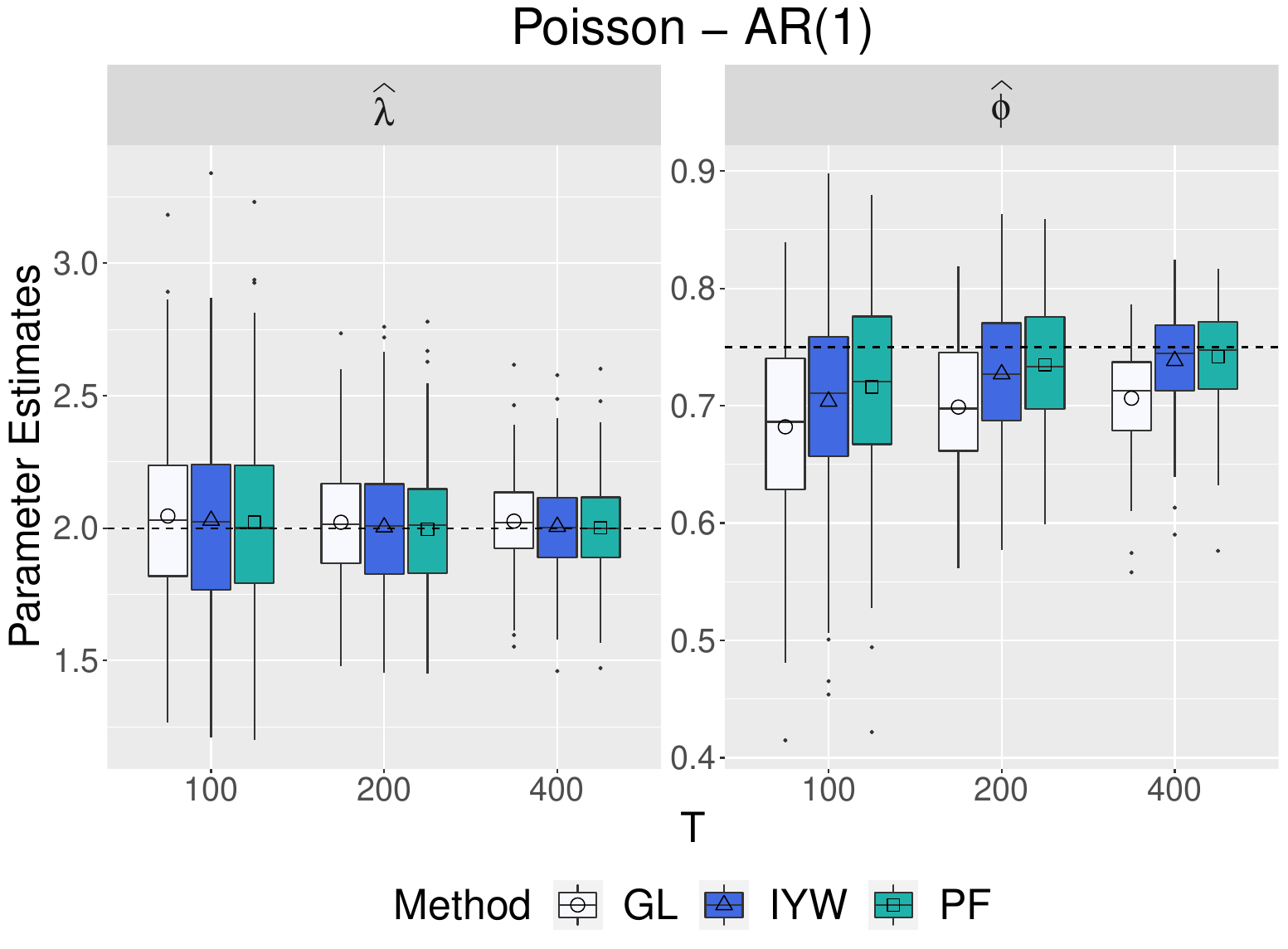}\quad
	\includegraphics[width = 0.49\textwidth]{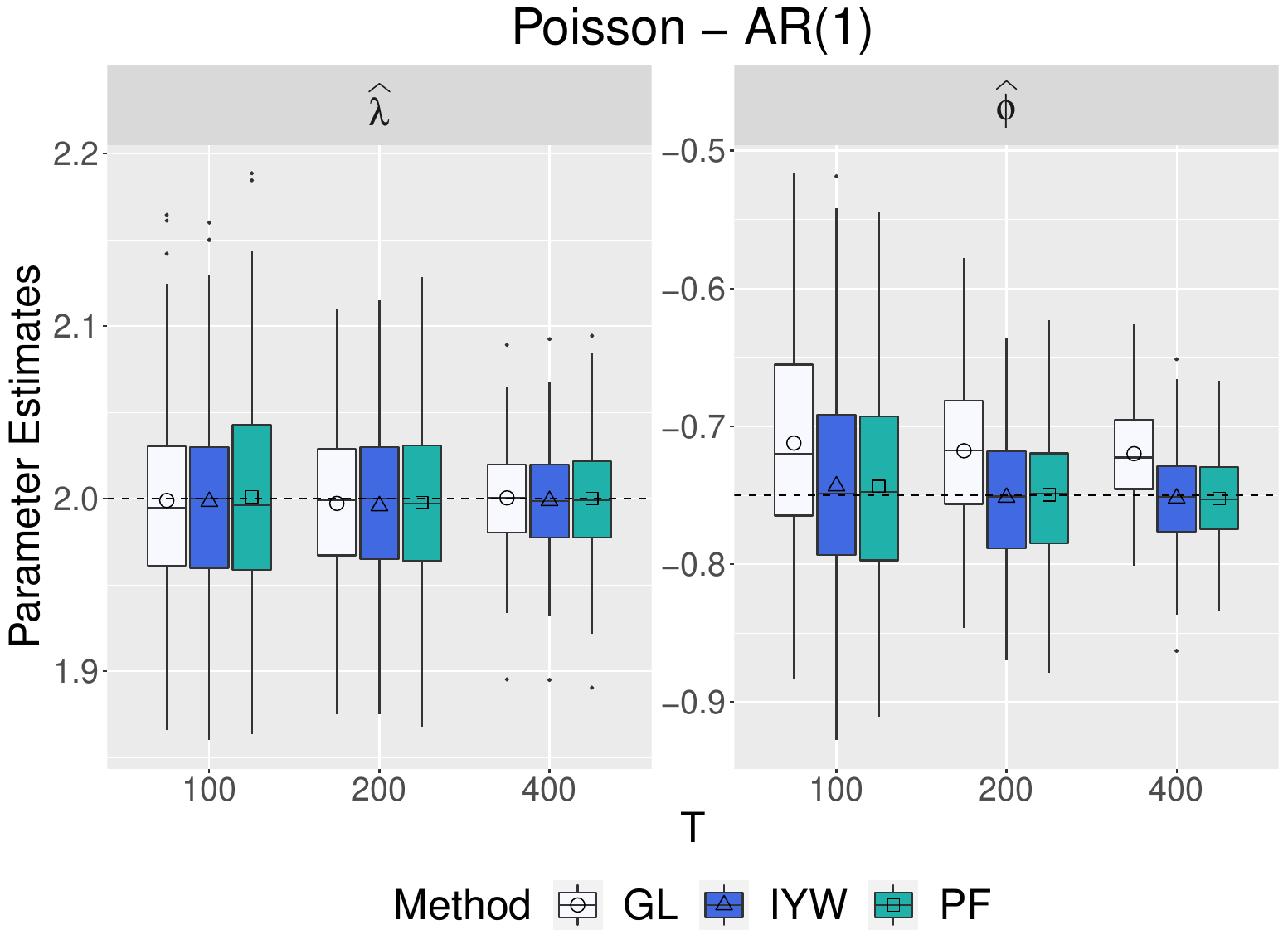}
	}
\caption{\label{f:pois-ar1-sim-lam2} {\small \textit{Gaussian likelihood, implied Yule-Walker, and PF/SMC
parameter estimates for 200 synthetic Poisson--AR(1) series with lengths $T = 100, 200$, and $400$.
The true parameter values (indicated by horizontal dashed lines) are $\lambda=2$ and
$\phi = 0.75$ (left panel), and $\lambda=2$ and $\phi=-0.75$ (right panel).}}}
\end{figure}

\section{An application}
\label{s:applications}

This section applies our methods to a weekly count series of product sales at Dominick’s Finer Foods, a now defunct U.S.~grocery chain that operated in Chicago, IL and adjacent areas from 1918 - 2013. Soft drink sales of an unnamed brand from a single store will be analyzed over a two-year span commencing on September 10, 1989. The series is plotted in Figure \ref{sales} (leftmost plot) and is part of a large and well-studied retail dataset, publicly available at https://www.chicagobooth.edu/research/kilts/datasets/dominicks (Source: The James M. Kilts Center for Marketing, University of Chicago).\footnote{In the dataset manual, the series in Figure \ref{sales} (leftmost plot) is the sales of the product with universal product code (UPC) 4640055081 from store 81.}.
\begin{figure}[h!]
	\includegraphics[width=0.45\textwidth,height=1.5in]{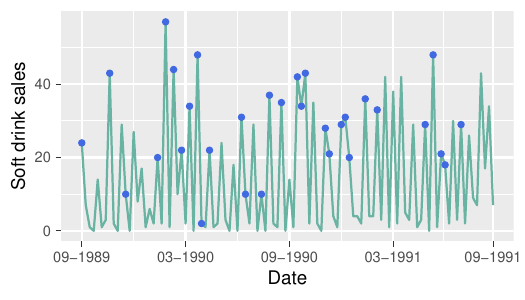}
	\includegraphics[width=0.26\textwidth,height=1.5in]{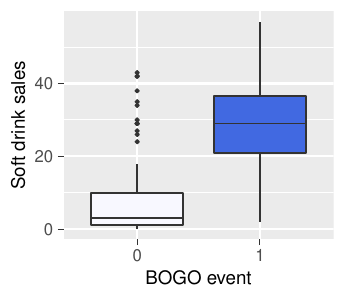}
	\includegraphics[width=0.26\textwidth,height=1.5in]{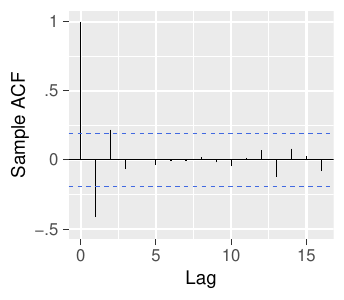}
	\caption{\label{sales} {\small \textit{Left: Weekly sales of a soft drink product sold at a single store of the grocery store
	\textit{Dominick's Finer Foods} from 09-10-1989 to 09-03-1991. The dots indicate the
	weekly sales were at least one ``Buy one and get one free'' (BOGO) sales promotion event took place.
	Middle: Boxplots of sales grouped by the BOGO covariate (0: weekly sales with no BOGO event, 1: weekly sales with at least
	one BOGO day during the week). Right: Sample ACF of the series with 95\% pointwise bands for zero correlation.}}}
\end{figure}
Our goal here  is not an in-depth retail analysis, but to illustrate our methods with a real world example of an overdispersed time series of small counts that has negative autocorrelation and dependence on a covariate.

The covariate we use is a zero-one ``buy one get one free'' (BOGO for short) sales promotion event $S_t$, $S_t=1$ implying that the BOGO promotion was offered at least one day during week $t$.  The dots in the left plot of Figure \ref{sales} signify that the week had at least one BOGO day. The middle plot shows the soft drinks sales distribution grouped by $S_t$, visually suggesting that a BOGO event increases soft drink sales. The rightmost plot shows the sample ACF of the series and reveals negative dependence at lag one. The lag one sample autocorrelation of the residuals after a linear regression of the series on the BOGO covariate is also negative, but comparatively smaller in magnitude.

To model overdispersion, negative binomial and generalized Poisson marginal distributions will be considered.  Although similar, these two distributions can yield different conclusions \citep{joe:2005}.  Following standard generalized linear modeling practice, both distributions are parametrized via the series' mean (although our setup allows covariates to enter through other parameters as well).  More specifically, for the negative binomial marginal, the standard pair $(r,p)$ used in Appendix C is now mapped to the parameter pair $(\mu, k)$, where $\mu = pr/(1-p)$ is the mean of the process and $k=1/r$ is the overdispersion parameter. Similarly, the generalized Poisson distribution of Appendix C is parametrized through the pair $(\mu,\alpha)$ as in \cite{famoye1993}, relation (2.4).   In this parametrization, $\mu$ is the mean of the series, whereas the sign of $\alpha$ controls the type of dispersion, with positive values indicating overdispersion.  To incorporate the BOGO covariate $S_t$ into the model, the mean of the series is allowed to depend on time $t$ through the typical GLM log-link $\mu_t=\exp \left( \beta_0 + \beta_1 S_{t} \right)$, while the parameters $k$ and $\alpha$ are kept fixed in time $t$.

An exploratory examination of the sample ACFs and PACFs of the series along with diagnostic plots of residuals obtained by fitting all ARMA$(p,q)$ models with $p,q \leq 5$ suggest an AR$(3)$ model as a suitable choice for $\{ Z_t \}$. Table 1 in Appendix D shows the AICc and BIC for both marginal distributions obtained via PF/SMC and GL methods (we omit IYW results for simplicity). The AR$(3)$ model was selected by AICc and BIC in both fits.  Interestingly, both the sample ACF and PACF of the series show one large non-zero value at lag one, but relatively smaller values at other lags (except perhaps the lag two value, which barely exceeds the 95\% $1.96/\sqrt{T}$ dashed confidence threshold for zero correlation).

We also considered a white noise latent series (labeled as ``WN'' Table 1 in Appendix D), which renders our model a standard GLM. The PF/SMC WN estimates from both distributions (omitted here for brevity) closely agree with parameter estimates obtained from exact generalized linear model fits (using, for example, functions from the R package ``MASS''). As expected, the WN model yielded the highest AICc and BIC values among all considered dependence structures, thus confirming the need for a model with temporal dependence.

Table \ref{NegBinAR3fit} shows parameter estimates and standard errors from fitting a negative binomial-AR(3). (Table 2 in Appendix D is for a generalized Poisson-AR(3) model.)  All marginal distributions and estimation methods yielded $\hat{\phi}_1<0$. Although a formal asymptotic theory is beyond the scope of our presentation here, asymptotic normality is expected.   Assuming this, the PF/SMC standard errors (the ones believed most trustworthy) suggest that all parameters are significantly non-zero at level $95\%$.  The findings suggest the negative binomial distribution is preferred over the generalized Poisson, that the correlation in the series at lag one is negative, and that a BOGO event indeed increases sales.

\begin{table}[h!]
	\centering
	{\small
	\scalebox{0.9}{
	\begin{tabular}{|c|cccccc|}
	\hline
	Parameters           & $\phi_1$ & $\phi_2$ & $\phi_3$ & $\beta_0$ & $\beta_1$  & $k$\\
	\hline
	GL Estimates         &   -0.447 &  0.145 &   0.208   & 2.433   &  0.569 &   0.884\\
	GL Standard Errors   &    0.175 &  0.171  &  0.130   & 0.095   &  0.115   & 0.207   \\
	\hline
	PF/SMC Estimates        &    -0.341& 0.223   & 0.291    &2.264   & 1.01  &  1.21 \\
	PF/SMC Standard Errors   &    0.100 & 0.107    &0.102   &0.142   &0.207   &  0.205\\
	\hline
	\end{tabular}}}
	\caption{\small \textit{{Estimates and standard errors of the negative binomial-AR(3) model.}}}
	\label{NegBinAR3fit}
\end{table}

Turning to residual diagnostics, the plots in Figure \ref{diagnostics} for the negative binomial-AR(3) fit suggest that the model has captured both the marginal distribution and the dependence structure. The residuals here were computed using \eqref{e:res}.

\begin{figure}[h!]
	\centering
	\includegraphics[height = 1.6in, width=2.1in]{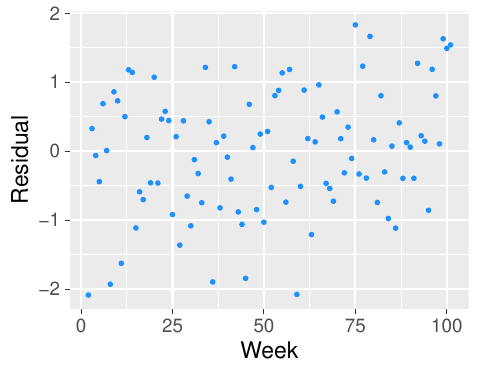}
	\includegraphics[height = 1.6in, width=2.1in]{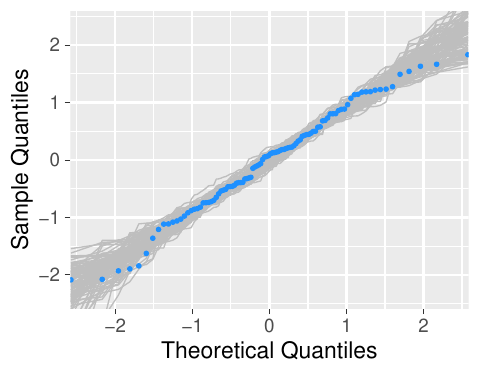}
	\includegraphics[height = 1.6in, width=2.1in]{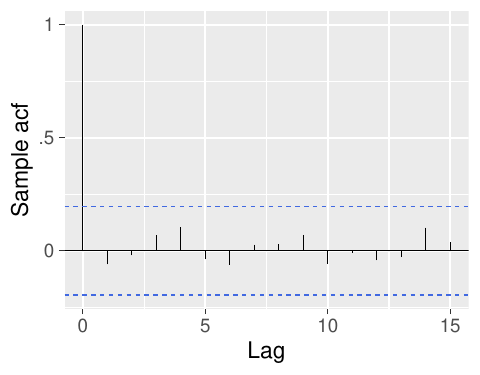}
	\caption{ {\small \textit{The leftmost plot displays the estimated residuals against time. The middle graph is a
	QQ plot for normality of the estimated residuals. The shaded region in the QQ plot shows100 realizations from
	a normal distribution with size, mean and standard deviation matching the residual sample counterparts.
	The right plot displays the sample autocorrelations of the estimated residuals.}}}
	\label{diagnostics}
\end{figure}

We next assess the predictive ability of the two fits via the non-randomized histograms shown in Figure \ref{PIT-2} and discussed in detail in Section \ref{s:inference-diagnostics}. We selected ten bins at the points $h/10, h=1, \ldots, 10$ as is typical in the literature. The negative binomial PIT plot suggests a satisfactory predictive ability with most bar heights being close to 0.1 (1 over the number of bins). In comparison, the generalized Poisson fit deviates more from the uniform distribution, with somewhat more pronounced peaks and valleys. We remind the reader here that PIT plots are known to be sensitive for smaller series lengths.  Quantifying this uncertainty (for each bin) through a statistical test is beyond the scope of this paper.  Nevertheless, we gauged the variability of the uniform distribution's bin heights through a small experiment. Specifically, 500 synthetic realizations of sample size $T=104$ were generated and the percentiles of all bin heights were collected. The 5th and 95th percentiles ranged in the intervals $(0.048, 0.058)$ and $(0.145, 0.154)$ respectively, suggesting that the peaks and valleys of the negative binomial PIT plot (which are within these percentiles) are mild; that is, uniformity is plausible and the marginal distribution fits seems adequate.

\begin{figure}[h!]
	\centering
	\includegraphics[ width=0.4\textwidth]{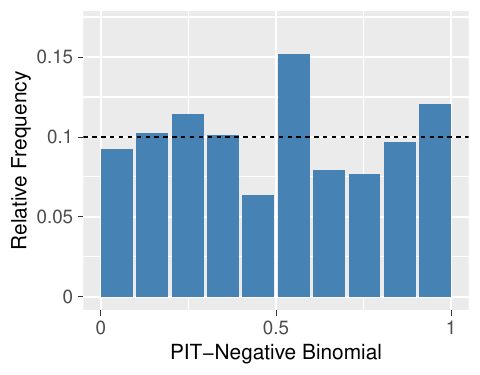}\quad
	\includegraphics[ width=0.4\textwidth]{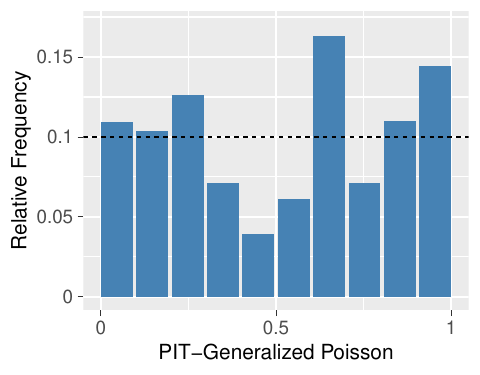}
	\caption{{\small\label{PIT-2} \textit{PIT residual histograms for the estimated models in Table \ref{NegBinAR3fit} and Table 2 in Appendix D.}
	}}
\end{figure}

\section{Conclusions and comments}
\label{s:conclusions}

This paper developed the theory and methods for a stationary count time series model made from a latent Gaussian process.  By using Hermite expansions, a very general model class was devised.  In particular, the autocorrelations in the series can be positive or negative, and in a pairwise sense, span the range of all achievable correlations.   The series can have any marginal distribution desired, thereby improving classical DARMA and INARMA count time series methods.   On inferential levels, autocovariances of the model were extracted from Hermite expansions, allowing for Gaussian pseudo-likelihood and implied Yule-Walker inference procedures.   A PF/SMC likelihood approach was also developed and produced estimators that were demonstrated to outperform the Gaussian pseudo-likelihood and implied Yule-Walker estimators in most cases.  These results complement the importance sampling methods for copula likelihoods in \cite{Smith2012}.  The methods were used in a simulation study and were applied in a regression analysis of a count series of weekly grocery sales that exhibited overdispersion, a negative lag one correlation, and dependence on a ``buy one get one free'' covariate. Model fits and predictive abilities of the methods were illustrated with generalized Poisson and negative binomial marginal distributions.

While the paper provides a reasonably complete treatment for count time series models, additional research is needed.  Some statistical issues, like asymptotic normality of parameter estimators, were not addressed here.  PF/SMC algorithms that optimize model likelihoods, which can be unwieldy, also merit further exploration.   The paper only considers univariate methods.   Multivariate count time series models akin to those in \cite{song2013} could be developed by replacing the univariate $\{ Z_t \}$ with a multivariate Gaussian process $\{ {\bf Z}_t \}$, whose components have a standard normal marginal distribution, but are cross-correlated for each fixed $t$.  The details for such a construction would proceed akin to the methods developed here.  Also, while the count case is considered here, the same methods will produce stationary time series having any general prescribed continuous distribution.  Finally, the same methods should prove useful in constructing spatial and spatio-temporal processes having any prescribed marginal distribution.  While \cite{DeOliveria_2013, HanDeOliveria_2016} recently addressed this issue in the spatial setting, additional work is needed, including exploring spatial Markov properties and likelihood evaluation techniques.   To the best of our knowledge, no comprehensive analogous work has been conducted for space-time count modeling to date.

\newpage
\begin{center}
	{\LARGE \textbf{Supplemental Material}}
\end{center}

There are four sections in this supplement. Sections \ref{a:further-coeff-link} and \ref{a:further-particle} contain further discussions, results, and proofs on Hermite coefficients, link functions, and particle filtering and state space methods. Section \ref{a:simulations} complements Section \ref{s:simulations} with additional simulation scenarios and Section \ref{a:application-tables} presents two auxiliary tables mentioned in the Section \ref{s:applications} application.
\appendix
\numberwithin{equation}{section}

\section{More on Hermite coefficients and link functions}
\label{a:further-coeff-link}

\subsection{Calculation and properties of the Hermite coefficients}
\label{s:model-herm-coef}
We first prove Lemma \ref{r:coeff-conv}.

\medskip

\noindent {\bf Proof of Lemma \ref{r:coeff-conv}:}
Recall that $\boldsymbol{\theta}$ denotes all parameters appearing in the marginal distribution $F_X$.  For $\boldsymbol{\theta}$ fixed, define the mass and cumulative probabilities of $F_X$ via
\begin{equation}
\label{e:Y-probs}
p_{n} = \PP[X_t = n], \quad C_{n} = \PP[ X_t \leq n] = \sum_{j=0}^n p_{j}, \quad n \in \{ 0,1,\ldots \},
\end{equation}
where dependence on $\boldsymbol{\theta}$ is notationally suppressed.  Note that
\begin{equation}
\label{e:Y-GX-G-again}
G(z) = \sum_{n=0}^\infty n\, 1_{\{ C_{n-1} \leq \Phi(z) <  C_{n}\}} = \sum_{n=0}^\infty n\, 1_{\big[\Phi^{-1}(C_{n-1}),\Phi^{-1}(C_{n})\big)}(z)
\end{equation}
(take $C_{-1}=0$ as a convention). When $C_n=0$, we take $\Phi^{-1}(C_n) = -\infty$; when $C_n=1$, we take $\Phi^{-1}(C_n) = \infty$. Using this in (\ref{e:herm-coef}) provides, for $k \geq 1$,
\begin{equation}
\label{e:herm-coeff}
g_{k} = \frac{1}{k!} \EE [G(Z_0) H_k(Z_0) ]= \frac{1}{k!} \sum_{n=0}^\infty n \EE \left[ 1_{\big[\Phi^{-1}(C_{n-1}),\Phi^{-1}(C_{n})\big)}(Z_0)H_k(Z_0) \right].
\end{equation}
Plugging the Hermite polynomials
\[
H_k(z) = (-1)^k e^{z^2/2} \frac{d^k}{dz^k} \left( e^{-z^2/2} \right), \quad z \in \RR,
\]
into (\ref{e:herm-coeff}) and simplifying provides
\begin{eqnarray}
g_k &=& \frac{1}{k!} \sum_{n=0}^\infty \frac{n}{\sqrt{2\pi}}
\int_{\Phi^{-1}(C_{n-1})}^{\Phi^{-1}(C_{n})} H_k(z) e^{-z^2/2} dz \nonumber \\
 &=&  \frac{1}{k!} \sum_{n=0}^\infty \frac{n}{\sqrt{2\pi}}
\int_{\Phi^{-1}(C_{n-1})}^{\Phi^{-1}(C_{n})} (-1)^k \Big( \frac{d^k}{dz^k} e^{-z^2/2}\Big) dz \nonumber \\
&=& \frac{1}{k!} \sum_{n=0}^\infty\frac{n}{\sqrt{2\pi}} (-1)^k \Big( \frac{d^{k-1}}{dz^{k-1}} e^{-z^2/2}\Big) \Big|_{z=\Phi^{-1}(C_{n-1})}^{\Phi^{-1}(C_{n})}  \nonumber \\
&=& \frac{1}{k!}\sum_{n=0}^\infty \frac{n}{\sqrt{2\pi}} (-1) e^{-z^2/2} H_{k-1}(z) \Big|_{z=\Phi^{-1}(C_{n-1})}^{\Phi^{-1}(C_{n})}  \nonumber \\
	& = & \frac{1}{k!\sqrt{2\pi}} \sum_{n=0}^\infty  n\, \Big[
	e^{-\Phi^{-1}(C_{n-1})^2/2} H_{k-1}(\Phi^{-1}(C_{n-1}))  - \nonumber \\
	& & \hspace{2.5in} e^{-\Phi^{-1}(C_{n})^2/2} H_{k-1}(\Phi^{-1}(C_{n}))
	\Big].\quad
\label{e:herm-coef-expr1}
\end{eqnarray}
The telescoping nature of the series in (\ref{e:herm-coef-expr1}) provides \eqref{e:herm-coef-expr2}.

Next, we discuss the convergence of this series.  Observe that one obtains (\ref{e:herm-coef-expr2}) from (\ref{e:herm-coef-expr1}) if, after changing $k-1$ to $k$ for notational simplicity,
\begin{equation}
\label{e:herm-coef-cond}
\sum_{n=0}^\infty e^{-\Phi^{-1}(C_{n})^2/2} \Big|H_k(\Phi^{-1}(C_{n})) \Big| <\infty.
\end{equation}	
To see that this holds when $\EE[ X_t^p ] < \infty$ for some $p > 1$, suppose that $C_n<1$ for all $n$, since otherwise the sum in (\ref{e:herm-coef-cond}) has a finite number of terms.  As $H_k(z)$ is a polynomial of degree $k$, $|H_k(z)|\leq \kappa(1+|z|^k)$ for some constant $\kappa$ that depends on $k$. The sum in (\ref{e:herm-coef-cond}) can hence be bounded (up to a constant) by
\begin{equation}
\label{e:bound-1}
\sum_{n=0}^\infty e^{-\Phi^{-1}(C_n)^2/2}(1+|\Phi^{-1}(C_n)|^k).
\end{equation}
To show that (\ref{e:bound-1}) converges, it suffices to show that
\begin{equation}
\label{e:bound-2}
\sum_{n=0}^\infty e^{-\Phi^{-1}(C_n)^2/2}|\Phi^{-1}(C_n)|^k < \infty
\end{equation}
since $|\Phi^{-1}(C_n)|^k \uparrow \infty$ as $C_n\uparrow 1$. Mill's ratio for a standard normal distribution states that $1-\Phi(x) \sim e^{-x^2/2}/(\sqrt{2\pi}x)$ as $x\to\infty$.  Substituting $x = \Phi^{-1}(y)$ gives $1-y\sim  e^{-\Phi^{-1}(y)^2/2}/(\sqrt{2\pi}\Phi^{-1}(y))$ as $y \uparrow 1$. Taking logarithms in the last relation and ignoring constant terms, order arguments show that $\Phi^{-1}(y) \sim \sqrt{2} |\log(1-y)|^{1/2}$ as $y \uparrow 1$. Substituting $\Phi^{-1}(C_n) \sim \sqrt{2} |\log(1-C_n)|^{1/2}$ into (\ref{e:bound-2}) provides
\begin{equation}
\label{e:herm-coef-cond-bound}
\sum_{n=0}^\infty e^{-\Phi^{-1}(C_n)^2/2}|\Phi^{-1}(C_n)|^k \leq \sum_{n=0}^\infty (1-C_{n}) |\log (1-C_n)|^{k/2} .
\end{equation}
	
For any $\delta>0$ and $x\in (0,1)$, one can verify that $-\log (x) \leq x^{-\delta}/\delta$. Using this in (\ref{e:herm-coef-cond-bound}) and $C_n=1-\PP[X>n]$, it suffices to prove that
\begin{equation}
\label{e:delta-inequality-2}
\sum_{n=0}^{\infty} \PP[X>n]^{1-\delta k /2} < \infty
\end{equation}
for some $\delta >0$. Since $X \geq 0$ and $\EE[X^p] < \infty$ is assumed, the Markov inequality gives $\PP[X>n] = \PP[X^p>n^p] \leq \EE[X^p]/n^p$. Thus, the sum in (\ref{e:delta-inequality-2}) is bounded by
\begin{equation}
\label{e:delta-inequality-3}
\EE[X^p]^{1-\delta k /2} \sum_{n=0}^{\infty} \frac{1}{n^{p-p\delta k/2}}.
\end{equation}
But (\ref{e:delta-inequality-3}) converges whenever $\delta < 2(p-1)/(pk)$. Choosing such a $\delta$ proves (\ref{e:herm-coef-cond}) and finishes our work. \quad\quad $\Box$

The following remarks and the next section shed light on the behavior of the Hermite coefficients in (\ref{e:herm-coef-expr2}).
	
\begin{Remark}
\label{r:coeff-calcul}
From a numerical standpoint, the expression in (\ref{e:herm-coef-expr2}) is evaluated as follows. The families of marginal distributions considered in this work have fairly ``light'' tails, meaning that $C_{n}$ approaches unity rapidly as $n \to \infty$.  This means that $C_{n}$ becomes {\it exactly unity numerically} for small to moderate values of $n$. Let $n(\boldsymbol{\theta})$ be the smallest such value. For example, for the Poisson distribution with parameter $\boldsymbol{\theta}=\lambda$ and Matlab software, $n(0.1) = 10$, $n(1) = 19$, and $n(10) = 47$. For $n \geq  n(\boldsymbol{\theta})$, the numerical value of $\Phi^{-1}(C_{n})$ is infinite and the terms $e^{-\Phi^{-1}(C_{n})^2/2} H_{k-1}(\Phi^{-1}(C_{n}))$ in (\ref{e:herm-coef-expr2}) are numerically zero and can be discarded. Thus, (\ref{e:herm-coef-expr2}) becomes
\begin{equation}
\label{e:herm-coef-expr-numer}
g_{k} = \frac{1}{k!\sqrt{2\pi}} \sum_{n=0}^{n(\boldsymbol{\theta})-1} e^{-\Phi^{-1}(C_{n})^2/2} H_{k-1}(\Phi^{-1}(C_{n})).
\end{equation}
Alternatively, one could calculate the Hermite coefficients using Gaussian quadrature methods, as discussed e.g.\ in \cite{HanDeOliveria_2016}, p.~51; however, the approach based on (\ref{e:herm-coef-expr-numer}) is numerically simpler. Furthermore, as noted below, the expression (\ref{e:herm-coef-expr-numer}) can shed further light on the behavior of the Hermite coefficients.
\end{Remark}
	
\begin{Remark}
\label{r:coeff-asympt-k-to-infinity}
Assuming that the $g_{k}$ are evaluated through (\ref{e:herm-coef-expr-numer}), their asymptotic behavior as $k \to \infty$ can be quantified. We focus on $g_{k}(k!)^{1/2}$, whose squares are the link coefficients. The asymptotic relation for Hermite polynomials states that $H_m(x) \sim e^{x^2/4} (m/e)^{m/2} \sqrt{2} \cos(x \sqrt{m} - m\pi/2)$ as $m \to \infty$ for each $x \in \RR$.  Using this and Stirling's formula ($k! \sim k^k e^{-k} \sqrt{2 \pi k}$ as $k \to \infty$) show that
\begin{equation}
\label{e:coeff-asympt-k-to-infinity}
g_{k} (k!)^{1/2} \sim \frac{1}{2^{1/4}\pi^{3/4}} \frac{1}{k^{3/4}}
\sum_{n=0}^{n(\boldsymbol{\theta})-1} e^{-\Phi^{-1}(C_{n})^2/4}
\cos\left(\Phi^{-1}(C_{n}) \sqrt{k-1} - \frac{(k-1)\pi}{2}\right).
\end{equation}
Numerically, this approximation, which does not involve Hermite polynomials, was found to be accurate for even moderate values of $k$. It implies that $k! g_{k}^2$ decays (up to a constant) as $k^{-3/2}$. While this might seem slow, these coefficients are multiplied by $\gamma_{Z}(h)^k=\rho_Z(h)^k$ in (\ref{e:Y-X-cov}), which decay geometrically in $k$ to zero, except in degenerate cases where $|\rho_Z(h)|=1$.
\end{Remark}
	
\begin{figure}[t!]
\centerline{
\includegraphics[width=3.25in,height=2.5in]{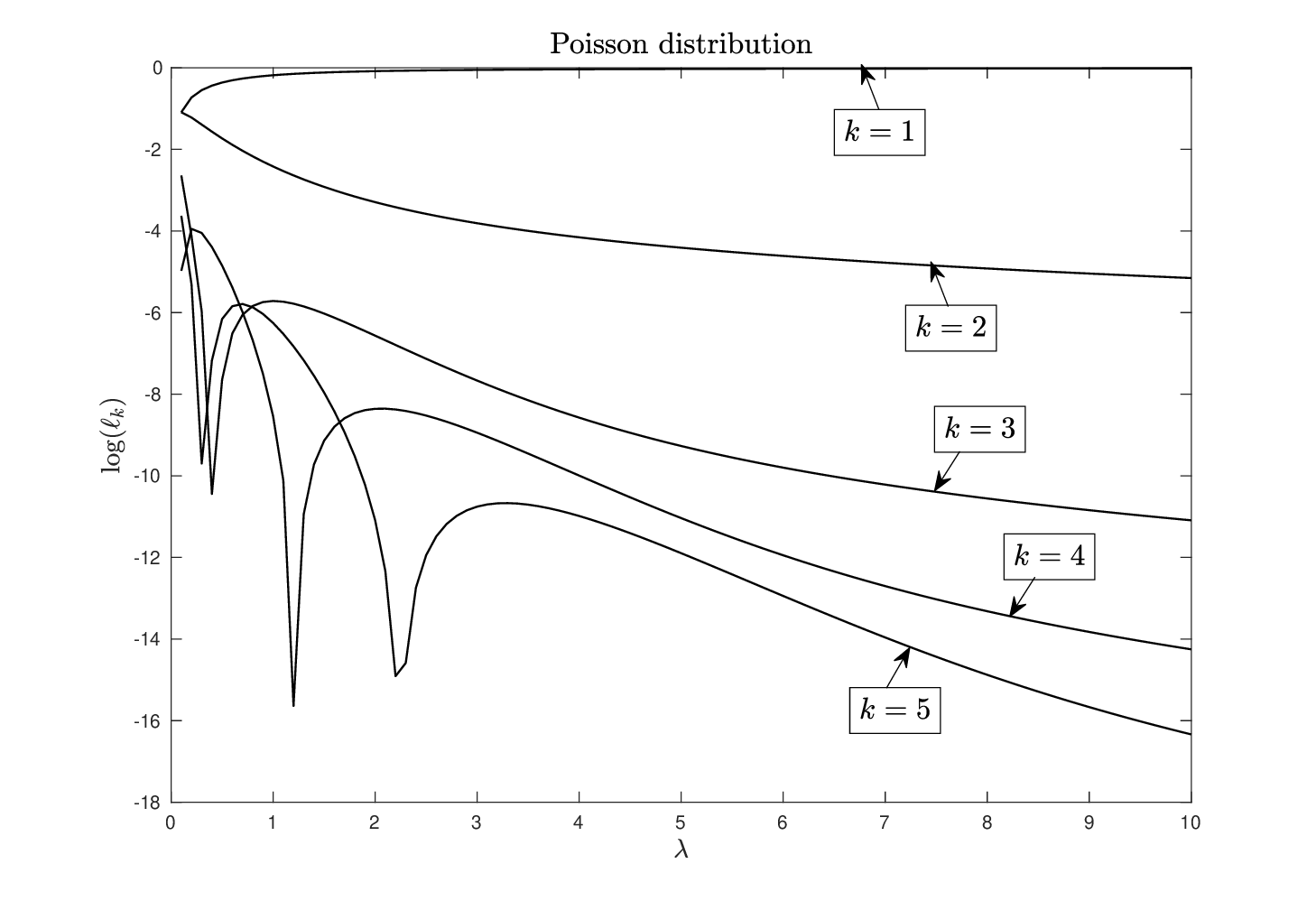}
\includegraphics[width=3.25in,height=2.5in]{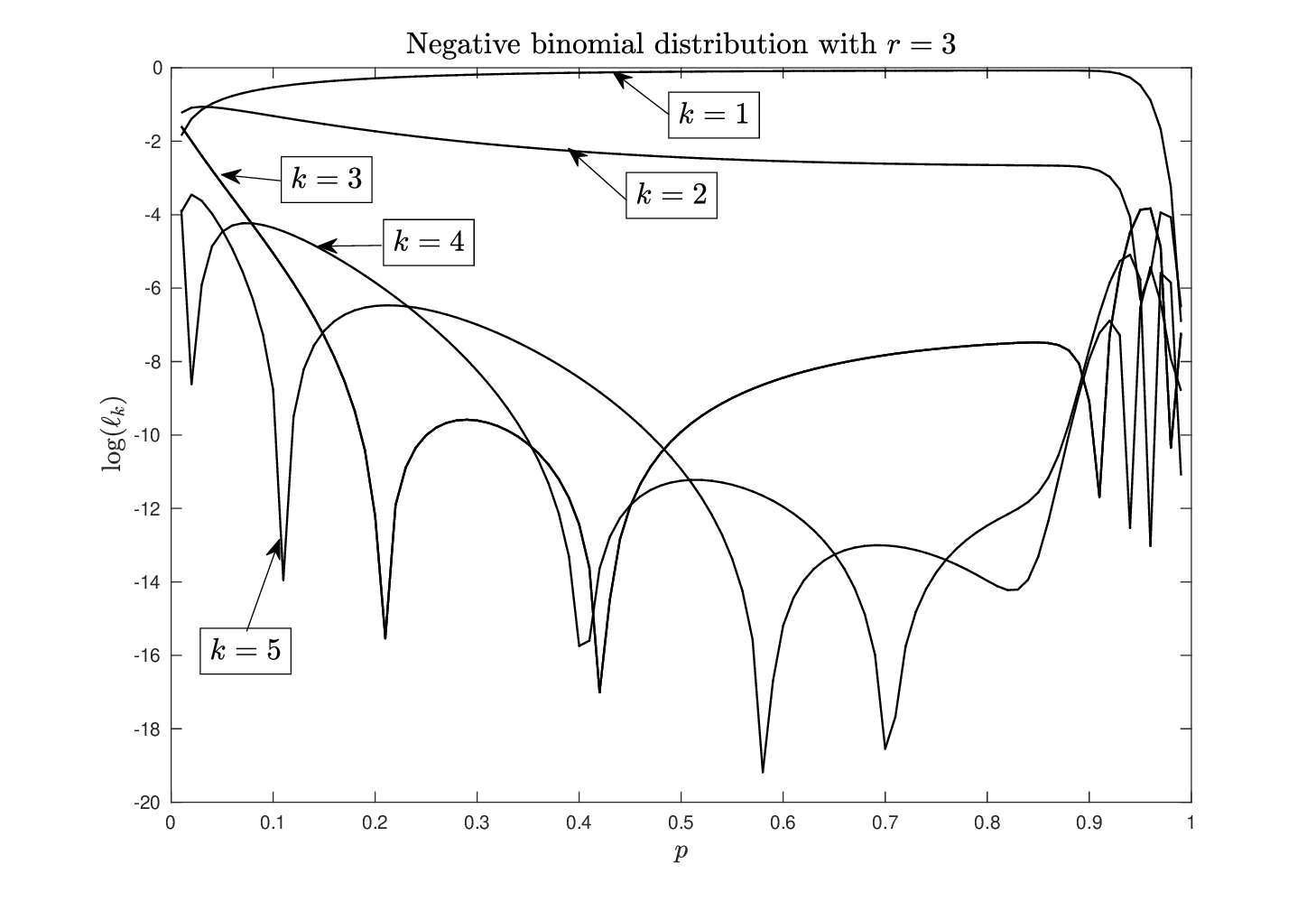}
}
	\caption{\label{f:coeffs-pois-nb} \textit{The link coefficients $\ell_{k}$ on a log-vertical scale for the
	Poisson (left) and negative binomial (right) distributions.}}
	\end{figure}
	
The computation and behavior of the link coefficients $\ell_{k}= k! g_k^2/\gamma_X(0)$ in (\ref{e:Y-X-cor-func}) are now examined for several families of marginal distributions (recalled in the beginning of  Section \ref{a:simulations}). Figure \ref{f:coeffs-pois-nb} shows plots of $\ell_{k}$ on a vertical log scale over a range of parameter values for $k\in\{1,\ldots,5\}$ for the Poisson and negative binomial (with $r=3$) distributions. A number of observations are worth making.
	
Since $\sum_{k=1}^\infty \ell_{k} = 1$ and $\ell_k \geq 0$ by construction, the parameter values in Figure \ref{f:coeffs-pois-nb} with $\log(\ell_{1})$ close to $0$ (or $\ell_{1}$ close to $1$) implies that most of the ``weight" in the link coefficients is contained in the first coefficient, with higher order coefficients being considerably smaller and decaying with increasing $k$. This takes place in the approximate ranges $\lambda>1$ for the Poisson distribution and $p\in (0.1,0.9)$ in the negative binomial distribution with $r=3$. Such cases will be called {\it ``condensed''}. As shown in Section \ref{s:model-link-func} below, $L(u)$ in the condensed case is close to $u$.   In the condensed case, correlations in $\{ Z_t \}$ and $\{ X_t \}$ are similar.
	
Non-condensed cases are referred to as {\it ``diffuse''}.  Here, weight is spread to many link coefficients. This happens in the approximate ranges $\lambda<1$ for the Poisson distribution and $p < 0.1$ and $p>0.9$ for the negative binomial distribution with $r=3$. This is expected for small $\lambda$s and small $p$s:  these cases correspond to discrete random structures that are nearly degenerate in the sense that they concentrate at $0$ (as $\lambda \to 0$ or $p \to 0$). For such cases, large negative correlations, such as $L(-1)$, are not possible; hence, $L(u)$ cannot be close to $u$ and correlations in $\{ Z_t \}$ and $\{ X_t \}$ are different. The diffuse range $p>0.9$ for the negative binomial distribution remains to be understood, although it is likely again some form of degeneracy.
	
\subsection{Calculation and properties of link functions}\label{s:model-link-func}
	
We now discuss calculation of $L(u)$ in (\ref{e:Y-X-cor-func}), which requires truncation of the sum to $k \in \{ 1, \ldots, K \}$ for some $K$. Note that the link coefficients $\ell_{k}$ are multiplied by $\gamma_{Z}(h)^k= \rho_Z(h)^k$ in (\ref{e:Y-X-cov}) before they are summed, and the latter decays to zero geometrically rapidly in $k$ for most stationary $\{ Z_t \}$ when $h \neq 0$. The link coefficients for large $k$ are therefore expected to play a minor role. We now set $K=25$ and explore consequences of this choice.
	
\begin{Remark}
An alternative procedure would bound (\ref{e:coeff-asympt-k-to-infinity}) by
\[
(2\pi^3 k^3)^{-1/4} \sum_{n=0}^{n(\boldsymbol{\theta})-1} e^{-\Phi^{-1}(C_{n})^2/4}.
\]
 Now let $K=K(\boldsymbol{\theta})$ be the smallest $k$ for which this bound is smaller than some preset error tolerance $\epsilon$. In the Poisson case with $\epsilon=0.01$, for example, such $K$ are $K(0.01)=29, K(0.1) = 27$, and $K(1) = 25$.  These are close to the chosen value of $K=25$. A different bound and resulting truncation in the spatial context can be found in \cite{HanDeOliveria_2016}, Lemma 2.2.
\end{Remark}
	
\begin{figure}[t!]
\centerline{
\includegraphics[width=3.25in,height=2.5in]{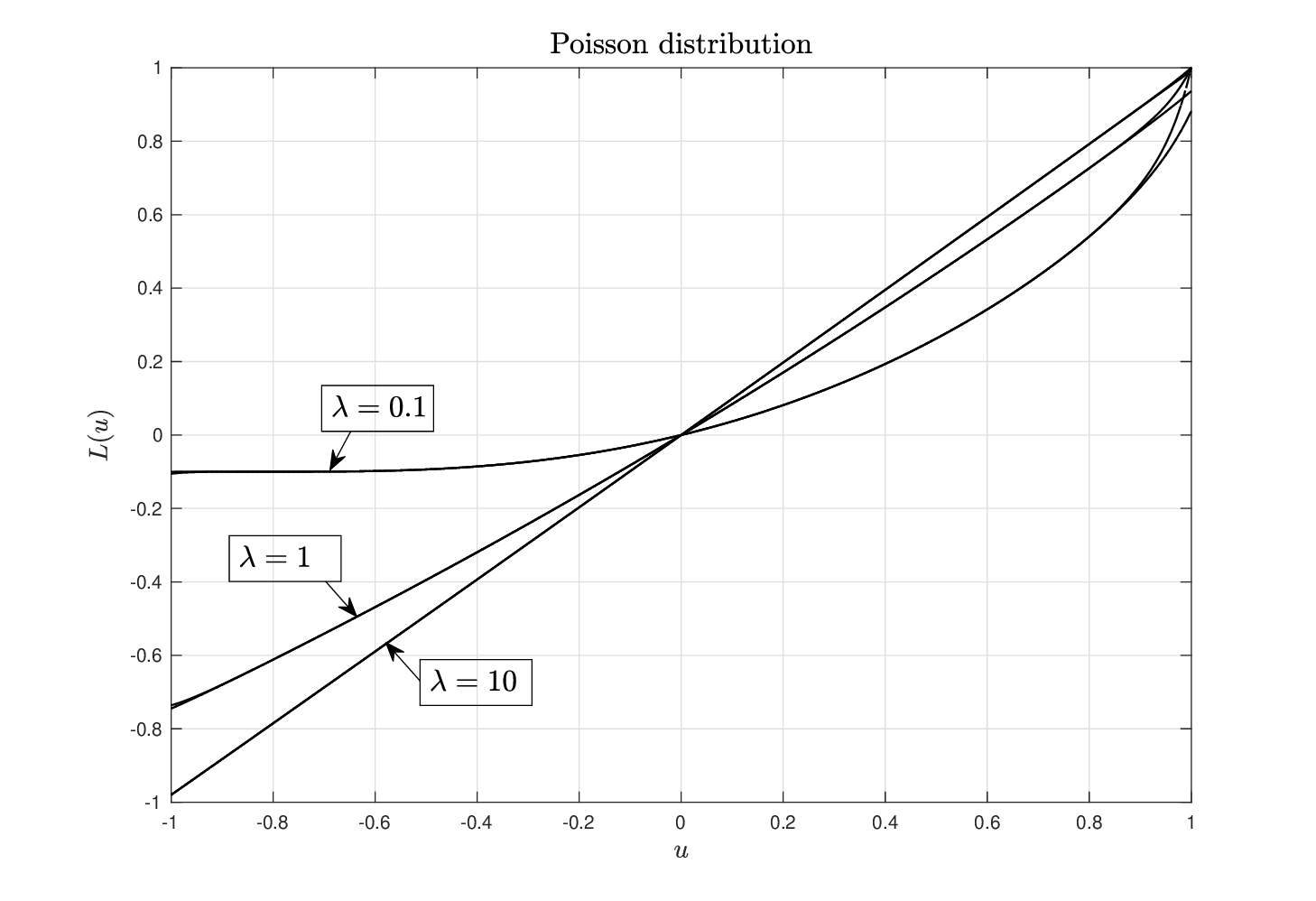}
\includegraphics[width=3.25in,height=2.5in]{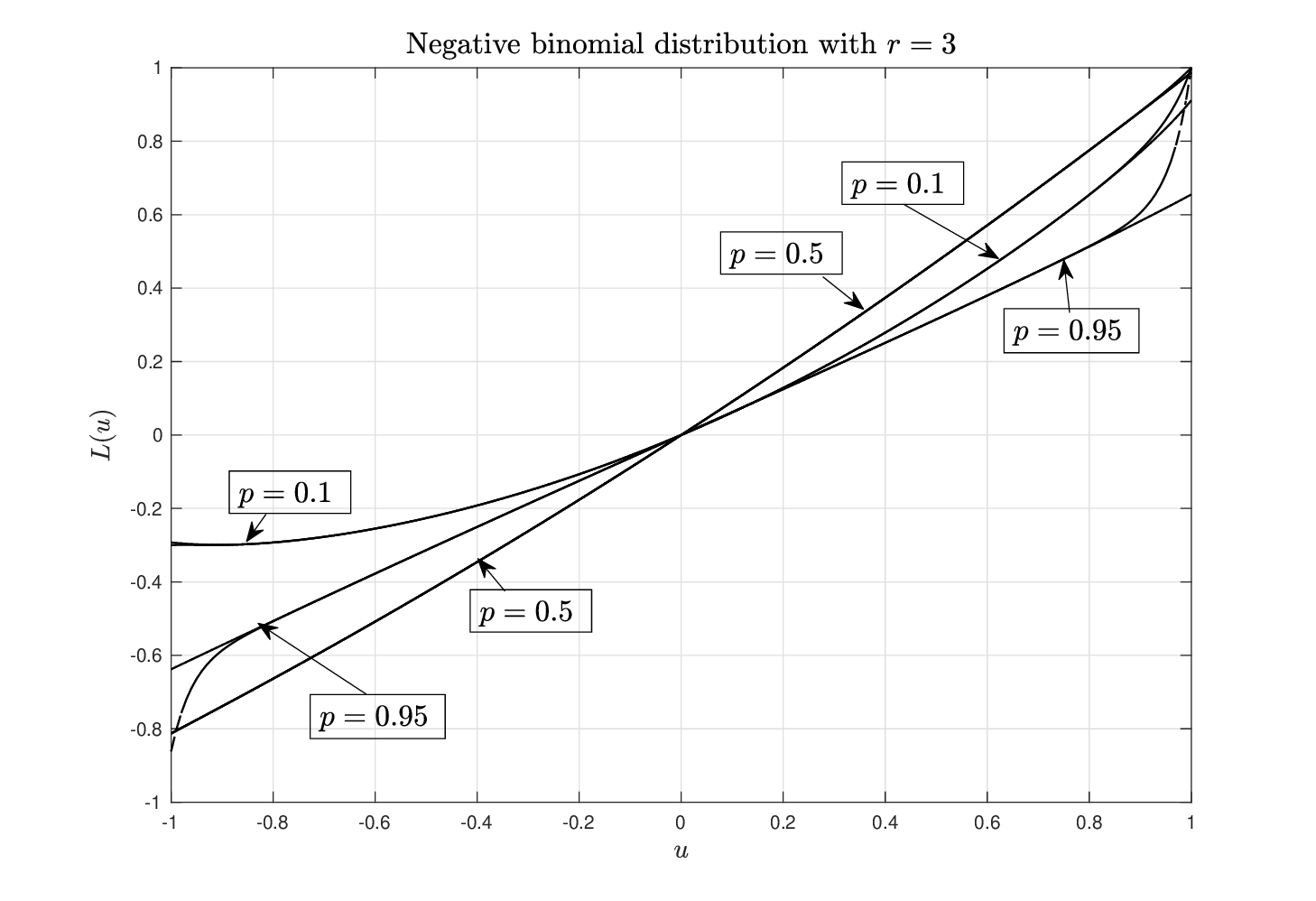}
}
\caption{\label{f:coeffs-pois-nb-hz} \textit{The link function $L$ for the Poisson distribution with
$\lambda=0.1$, $1$, and $10$ (left) and the negative binomial distribution with $r=3$ and $p=0.1$, $0.5$,
and $0.95$ (right).}}
\end{figure}
	
Figure \ref{f:coeffs-pois-nb-hz} plots $L(u)$ (solid line) for the Poisson and negative binomial distributions for several parameter values. The link function is computed by truncating its expansion to $k \leq 25$ as discussed above. The condensed cases $\lambda=10$ and $\lambda=1$ (perhaps this case is less condensed) and $p=0.85$ lead to curves with $L(u) \approx u$.  However, the diffuse cases appear more delicate. Diffusivity and truncation of the infinite series in (\ref{e:Y-X-cor-func}) lead to a computed link function that does not have $L(1)=1$ (see Section \ref{s:model-rel-acvfs}); in this case, one should increase the number of terms in the summation.
	
Though deviations from $L(1)=1$ might seem large (most notably for the negative binomial distribution with $p=0.95$), this seems to arise only in the more degenerate cases associated with diffusivity; moreover, this occurs only when linking an ACF of $\{ Z_t \}$ for lags $h$ for which $\rho_{Z}(h)$ is close to unity.  For example, note that if the link deviation is $0.2$ from unity at $u=1$ (as it is approximately for the negative binomial distribution with $p=0.95$), the error for linking $\rho_{Z}(h)$ as 0.8 (or smaller but positive) would be no more than $0.2 (0.8)^{26} \approx 0.0006$. In practice, any link deviation could be partially corrected by adding one extra ``pseudo link coefficient", in our case, a 26th coefficient, which would make the link function pass through $(1,1)$. The resulting link function is depicted in the dashed line in Figure \ref{f:coeffs-pois-nb-hz} around the point $(1,1)$ and essentially coincides with the original link function for all $u$'s except possibly for $u$ values that are close to unity.
	
The situation for negative $u$ and, in particular, around $u=-1$ is different:  the theoretical value of $L(-1)$ in Section \ref{s:model-rel-acvfs} is not explicitly known.  However, a similar correction could be achieved by first estimating $L(-1)$ through a Monte-Carlo simulation and adding a pseudo 26th coefficient making the computed link function connect to the desired value at $u=-1$. This is again depicted for negative $u$ via the dashed lines in Figure \ref{f:coeffs-pois-nb-hz}, which is visually distinguishable only near $u=-1$ (and then only in some cases).  Again, one cannot have a count series whose lag one correlation is more negative than $L(-1)$ --- such a count series does not exist by Remark \ref{r:max-min-corr}.
	
\begin{Remark}
In our estimation work, a link function needs to be evaluated multiple times; hence, running Monte-Carlo simulations to evaluate $L(-1)$ can become computationally expensive. In this case, the estimation procedure is fed precomputed values of $L(-1)$ on a grid of parameter values and interpolation is used for intermediate parameter values.
\end{Remark}
	
The next result further quantifies the link function's structure. The result implies that $\rho_X(h)$ is nondecreasing as a function of $\rho_Z(h)$.  The link's strict monotonicity is known from \cite{Grigoriu_2007} when $G$ is non-decreasing and differentiable, which does not hold in our case. (Non-strict) monotonicity for arbitrary non-decreasing $G$ is also argued in \cite{Cairo_Nelson_1997}. Our argument extends strict monotonicity to our setting and identifies an explicit form for the link function's derivative.
	
\begin{Proposition}
\label{p:link-derivative}
Let $L(\cdot)$ be the link function in (\ref{e:Y-X-cor-func}). Then, for $u \in (-1,1)$,
\begin{equation}
\label{e:link-derivative}
L'(u) = \frac{1}{2\pi \gamma_X(0) \sqrt{1-u^2}} \sum_{n_0=0}^\infty \sum_{n_1=0}^\infty e^{-\frac{1}{2(1-u^2)}\big( \Phi^{-1}(C_{n_0})^2 + \Phi^{-1}(C_{n_1})^2 - 2 u \Phi^{-1}(C_{n_0}) \Phi^{-1}(C_{n_1})\big)}.
\end{equation}
In particular, $L(u)$ is monotone increasing for $u \in (-1,1)$.
\end{Proposition}
	
\noindent {\bf Proof of Proposition \ref{p:link-derivative}:} We first derive the expression (\ref{e:link-derivative}) informally and then furnish the technicalities. When $G(\cdot)$ in (\ref{e:Y-GX}) is continuous and differentiable, the derivative of the link function can be obtained from the Price Theorem (Theorem 5.8.5 in \cite{pipiras:2017}); namely, for $u\in (-1,1)$,
\begin{equation}
\label{e:price}
L'(u) = \frac{1}{\gamma_X(0)} \EE [ G'(Z_0) G'(Z_1) ] \Big|_{\mbox{\scriptsize Corr}(Z_0,Z_1)=u}
\end{equation}
(the notation indicates that the correlation between the standard Gaussian pair $(Z_0,Z_1)$ is $u$).   If $G$ is further nondecreasing, then $G^\prime(x) \geq 0$ for all $x$ and (\ref{e:price}) implies that $L^\prime (u) \geq 0$ for all $u$.  This is the argument in \cite{Grigoriu_2007}.  While our $G$ is nondecreasing, it can be seen to be piecewise constant via (\ref{e:Y-GX-G-again}) and is hence not differentiable at its jump points.

To remedy this, we approximate $G$ by differentiable functions, apply  (\ref{e:price}), and take limits in the approximation error. Executing on this, for $\epsilon >0$ and $U\stackrel{\cal D}{=}{\cal N}(0,1)$, set
\begin{eqnarray}
G_\epsilon(x) & = & \EE [ G(x + \epsilon U)] = \int_{- \infty}^\infty G(z)
\frac{e^{-\frac{(x-z)^2}{2\epsilon^2}}}{\sqrt{2 \pi} \epsilon} dz \nonumber \\
& = & \sum_{n=0}^\infty n \int_{\Phi^{-1}(C_{n-1})}^{\Phi^{-1}(C_n)}
\frac{e^{-\frac{(x-z)^2}{2\epsilon^2}}}{\sqrt{2 \pi} \epsilon} dz \nonumber \\
& = & \sum_{n=0}^\infty n \int_{\Phi^{-1}(C_{n-1})-x}^{\Phi^{-1}(C_n)-x}
	\frac{e^{-\frac{w^2}{2\epsilon^2}}}{\sqrt{2 \pi}\epsilon} dw,
\label{e:G-epsilon}
\end{eqnarray}
where the expression in (\ref{e:Y-GX-G-again}) was substituted for $G(z)$. As $\epsilon \downarrow 0$, $G_\epsilon(x)$ approximates $G(x)$ since the ``kernel'' $e^{-\frac{(x-z)^2}{2\epsilon^2}}/(\sqrt{2\pi}\epsilon)$ acts like Dirac's delta function $\delta_{\{x\}}(z)$ at $z=x$.  Let $L_\epsilon$ be the link function induced by $G_\epsilon$, and $X_t^{(\epsilon)} = G_\epsilon(Z_t)$ its corresponding time series. Observe that $G_\epsilon(x)$ is nondecreasing and is differentiable by (\ref{e:G-epsilon}) with derivative
\begin{equation}
\label{e:G-epsilon-derivative}
G_\epsilon'(x) = \frac{1}{\sqrt{2\pi}\epsilon} \sum_{n=0}^\infty n \Big[ e^{-\frac{(\Phi^{-1}(C_{n-1})-x)^2}{2\epsilon^2}} -  e^{-\frac{(\Phi^{-1}(C_{n})-x)^2}{2\epsilon^2}} \Big] = \frac{1}{\sqrt{2\pi}\epsilon} \sum_{n=0}^\infty e^{-\frac{(\Phi^{-1}(C_{n})-x)^2}{2\epsilon^2}},
\end{equation}
where the last step uses the telescoping nature of the series, justifiable from the finiteness of  $\EE[ X_t^p]$ for some $p > 1$ analogously to (\ref{e:herm-coef-expr1}) and (\ref{e:herm-coef-expr2}). Formula (\ref{e:price}) now yields
\begin{eqnarray}
L_\epsilon'(u) & = & \frac{1}{\gamma_{X^{(\epsilon)}}(0)} \EE [G_\epsilon'(Z_0) G_\epsilon'(Z_1) ] \Big|_{\mbox{\scriptsize Corr}(Z_0,Z_1)=u} \nonumber \\
& = & \frac{1}{\gamma_{X^{(\epsilon)}}(0)}
\int_{-\infty}^\infty
\int_{-\infty}^\infty G_\epsilon'(z_0) G_\epsilon'(z_1) \frac{1}{2\pi\sqrt{1-u^2}} e^{-\frac{1}{2(1-u^2)}\big(z_0^2 + z_1^2 - 2 u z_0 z_1\big)} dz_0dz_1 \nonumber \\
& = & \frac{1}{\gamma_{X^{(\epsilon)}}(0)} \sum_{n_0=0}^\infty \sum_{n_1=0}^\infty \int_{-\infty}^\infty  \int_{- \infty}^\infty
\frac{1}{\sqrt{2\pi}\epsilon} e^{-\frac{(\Phi^{-1}(C_{n_0})-z_0)^2}{2\epsilon^2}}
\frac{1}{\sqrt{2\pi}\epsilon} e^{-\frac{(\Phi^{-1}(C_{n_1})-z_1)^2}{2\epsilon^2}} \times \nonumber \\
& & \quad \quad  \ \frac{1}{2\pi\sqrt{1-u^2}} e^{-\frac{1}{2(1-u^2)}\big(z_0^2 + z_1^2 - 2 u z_0 z_1\big)} dz_0dz_1.
\end{eqnarray}
Noting again that  $e^{-\frac{(x-z)^2}{2\epsilon^2}}/(\sqrt{2\pi}\epsilon)$ acts like a Dirac's delta function $\delta_{\{x\}}(z)$, the limit as $\epsilon \downarrow 0$ should be
\begin{equation}
\label{e:link-derivative-again}
L'(u) = \frac{1}{\gamma_X(0)}  \sum_{n_0=0}^\infty \sum_{n_1=0}^\infty \frac{1}{2\pi \sqrt{1-u^2}}
e^{-\frac{1}{2(1-u^2)}\big( \Phi^{-1}(C_{n_0})^2 + \Phi^{-1}(C_{n_1})^2 - 2 u \Phi^{-1}(C_{n_0}) \Phi^{-1}(C_{n_1})\big)},
\end{equation}
which is (\ref{e:link-derivative}) and is always non-negative.   Note that the derivative of $L$ always exists in $u \in (-1,1)$ since $L(u)$ is a power series with positive coefficients that sum to unity.

The formal justification of (\ref{e:link-derivative-again}) proceeds as follows. We focus only on proving that $L_\epsilon'(u)$ converges to $L'(u)$, which is the most difficult step. For this, we first need an expression for the Hermite coefficients of $G_\epsilon(\cdot)$, denoted by $g_{\epsilon,k}$.   These will be compared to the Hermite coefficients $g_k$ of $G$. Using $H_k(x+y) = \sum_{\ell=0}^k {k \choose \ell} y^{k-\ell} H_\ell(x)$, note that
\begin{eqnarray*}
G_\epsilon(z) & = &  \EE [G(x+\epsilon U)] =
\EE \left[ \sum_{k=0}^\infty g_k H_k(x+\epsilon U) \right] \\
& = &  \EE \left[ \sum_{k=0}^\infty g_k \sum_{\ell=0}^k {k \choose \ell} (\epsilon U)^{k-\ell} H_\ell(x) \right] \\
& = & \sum_{\ell=0}^\infty H_\ell(x) \sum_{k=\ell}^\infty g_k \epsilon^{k-\ell} {k \choose \ell}  \EE[U^{k-\ell}].
\end{eqnarray*}
Then, after changing summation indices and using that $\EE[U^p] = 0$ if $p$ is odd, and equal to $(p-1)!!$ if $p$ is even, where $k!!=1 \times 3 \times \cdots \times k$ when $k$ is odd, we get
\begin{equation}
\label{e:gk-epsilon}
g_{\epsilon,k} = g_k + \sum_{q=1}^\infty g_{k+2q} \epsilon^{2q} {k+2q \choose k} (2q-1)!! =  g_k + \sum_{q=1}^\infty g_{k+2q} \epsilon^{2q}\frac{(k+2q)!}{k!2^qq!}.
\end{equation}
This implies that
\begin{equation}
\label{e:gk-epsilon-gk}
|g_k^2 - g_{k,\epsilon}^2| \leq 2|g_k| \sum_{q=1}^\infty |g_{k+2q}| \epsilon^{2q}\frac{(k+2q)!}{k!2^qq!} + \Big(\sum_{q=1}^\infty |g_{k+2q}| \epsilon^{2q}\frac{(k+2q)!}{k!2^qq!} \Big)^2.
\end{equation}
The Cauchy-Schwarz inequality gives the bound
$$
\sum_{q=1}^\infty |g_{k+2q}| \epsilon^{2q}\frac{(k+2q)!}{k!2^qq!} \leq
\left( \sum_{q=1}^\infty g_{k+2q}^2 (k+2q)! \right)^{1/2}
\left( \sum_{q=1}^\infty \epsilon^{4q} \frac{(k+2q)!}{(k!)^2(2^qq!)^2} \right)^{1/2}
$$
$$
\leq \frac{M_k}{(k!)^{1/2}}
\left(  \sum_{q=1}^\infty \epsilon^{4q} \frac{(k+2q)!}{k!(2q)!} \right)^{1/2},
$$
where $M_k$ is some finite constant that converges to zero as $k \rightarrow \infty$.   Here, we have used that $\sum_{q=1}^\infty g_{k+2q}^2 (k+2q)! \to 0$ as $k \rightarrow \infty$, which is justifiable from $\mbox{\rm Var}(X_t) = \gamma_{X}(0) = \sum_{k=1}^\infty k! g_{k}^2$, and the fact that $(2^qq!)^2$ is of the same order as $(2q)!$. To bound sums of the form $\sum_{p=1}^\infty \epsilon^{2p} {k+p \choose p}$, use $\sum_{p=0}^\infty x^{p} {k+p\choose p} = (1-x)^{-k-1}$, $|x|<1$.   Collecting the above bounds and returning to (\ref{e:gk-epsilon-gk}) gives
\begin{equation}
\label{e:gk-epsilon-gk-2}
|g_k^2 - g_{k,\epsilon}^2| \leq  \frac{2 M_k|g_k|}{(k!)^{1/2}} \left[ (1-\epsilon^2)^{-k-1} -1 \right]^{1/2} + \frac{M_k^2}{k!} [(1-\epsilon^2)^{-k-1} -1].
\end{equation}

The rest of the argument is straightforward with this bound; in particular, note from (\ref{e:Y-X-cor-func}) that
$$
L'(u) = \sum_{k=1}^\infty \frac{g_k^2k!}{\gamma_X(0)} ku^{k-1},
$$
where the series converges for $u\in (-1,1)$ since the ``extra'' $k$ gets dominated by $u^{k-1}$. Similarly,
$$
L_\epsilon'(u) = \sum_{k=1}^\infty \frac{g_{\epsilon,k}^2k!}{\gamma_{X^{(\epsilon)}}(0)} ku^{k-1}.
$$
Then,
\begin{equation}\label{e:L-Lepsilon}
|L'(u) - L_\epsilon'(u)| \leq \Big| \frac{1}{\gamma_X(0)} - \frac{1}{\gamma_{X^{(\epsilon)}}(0)} \Big|  \sum_{k=1}^\infty g_k^2k!  k|u|^{k-1} + \frac{1}{\gamma_{X^{(\epsilon)}}(0)} \sum_{k=1}^\infty |g_k^2 - g_{\epsilon,k}^2|k! k|u|^{k-1}.
\end{equation}
For example, the series in the last bound converges to $0$ as $\epsilon \downarrow 0$. Indeed, by using (\ref{e:gk-epsilon-gk-2}), this follows if
$$
\sum_{k=1}^\infty |g_k| (k!)^{1/2} \left[ (1-\epsilon^2)^{-k-1} -1 \right]^{1/2} k|u|^{k-1} \to 0,\quad
\sum_{k=1}^\infty [(1-\epsilon^2)^{-k-1} -1] k|u|^{k-1} \to 0.
$$
In both of these cases, the convergence follows from the dominated convergence theorem since $(1-\epsilon^2)^{-k-1} -1\to 0$ as $\epsilon \downarrow 0$. By using $\mbox{\rm Var}(X_t) = \gamma_{X}(0) = \sum_{k=1}^\infty k! g_{k}^2$, one can similarly show that $\gamma_{X^{(\epsilon)}}(0)\to \gamma_X(0)$. Hence, we conclude that $L_\epsilon'(u)\to L'(u)$ as $\epsilon \downarrow 0$. \quad\quad $\Box$

\begin{Remark}
The antiderivative
$$
\int
\frac{\exp{ \left[ - \frac{a^2+b^2 -2uab}{2(1-u^2)}\right] }}
{ \sqrt{1-u^2}} du
$$
does not seem to have a closed form expression for general $a,b \in \RR$. (If it did, then one could integrate (\ref{e:link-derivative}) explicitly and get a closed form expression for $L(u)$.) But a number of numerical ways to evaluate the above integral over a finite interval have been studied; see, for example, \cite{genz:2004numerical}, Section 2.
\end{Remark}

\section{Particle filtering and sequential Monte Carlo methods}
\label{a:further-particle}

The next three remarks connect our model and sequential importance sampling (SIS) algorithm from Section \ref{s:particles} to the GHK sampler, hidden Markov models (HMMs), state space models (SSMs), and particle filetering (PF) and sequential Monte Carlo methods (SMCs).
\begin{Remark}
\label{r:GHK}
Note that, by using \eqref{e:truncation set}--\eqref{e:truncation set-2},
\begin{eqnarray}
\PP[X_0=x_0,\ldots , X_t=x_T] &=& \PP[Z_0 \in A_{x_0},\ldots, Z_T \in A_{x_T}] \nonumber \\
&=& \int_{ \{ A_{x_s}, s=0, \ldots, T \}} \frac{ e^{-\frac{1}{2}\sum_{s=0}^{T} (z_s - \hat{z}_s)^2/r_s^2 }}
{(2\pi)^{(T+1)/2} r_0 \ldots r_{T}} dz_0 \ldots dz_{T}. \quad \label{e:integral}
\end{eqnarray}
By \eqref{e:weight-expectation}, the truncated integral \eqref{e:integral} over a multivariate Gaussian density is (up to $\PP[X_0=x_0]$) equal to $\mathbb{E}[w_t^i]$, which by using SIS particles, is approximated by the sample average of
$$
w_T^i
= \prod_{t=1}^T w_t(\hat Z_t^i)
=  \prod_{t=1}^T
\left[
\Phi
\left( \frac{\Phi^{-1}(C_{x_t})-\hat Z_t^i}{r_t}\right)
-
\Phi
\left( \frac{ \Phi^{-1}(C_{x_t-1})-\hat Z_t^i}{r_t}\right)
\right].
$$
Using the sample average of $w_T^i$ to approximate the truncated multivariate integral \eqref{e:integral}, the SIS procedure can be viewed as the popular GHK sampler (\citep{hajivassiliou:ruud:1994}, p.\ 2405). Our contribution is to note that the likelihood can be expressed through the normal integral \eqref{e:integral}, involving one-step-ahead predictions and their errors, that can be efficiently computed through standard techniques from the time series literature. The GHK sampler is also used in \cite{masarotto:2012}, p.\ 1528, and \cite{HanDeOliveria_2018}.
\end{Remark}

\begin{Remark}
\label{r:HMM}
When $\{ Z_t \}$ is an AR$(p)$, $(Z_t, \ldots, Z_{t-p+1})'$ is a Markov chain on $\RR^p$, $\{ X_t\}$ defined by $ X_t= G(Z_t)$ is a SSM  or HMM (the same conclusion applies to ARMA$(p,q)$ models with an appropriate state space enlargement).  Indeed, when $p=1$, the AR$(1)$ model with a unit variance can be written as $Z_t = \phi Z_{t-1} + (1-\phi)^{1/2}\epsilon_t$, where $|\phi|<1$ and $\{\epsilon_t\}$ consists of IID ${\cal N}(0,1)$ random variables. Then $\{ X_t \}$ is an HMM in the sense of Definition 9.3 in \cite{douc:2014} with a Markov kernel on $\RR$ of
\begin{equation}
\label{e:ar1-kernelX}
M(z,dz') =
\frac{e^{-\frac{(z'-\phi z)^2}{2(1-\phi^2)}} dz}
{\sqrt{2\pi (1-\phi^2)}}
\end{equation}
governing transitions of $\{ Z_t \}$, and a Markov kernel from $\RR$ to $\NN_0$, serving as the state equation, of
\begin{equation}
\label{e:ar1-kernelXY}
G(z,dx) = \delta_{G(z)}(dx) = \mbox{point mass at}\ G(z)
\end{equation} 	
governing the connection between $Z_t$ and $X_t$. Thus, many HMM developments (see e.g.\ Chapters 9--13 in \cite{douc:2014}) apply to our model for Gaussian AR($p$) $\{ Z_t \}$. One important feature of our model when viewed as an HMM is that it is {\it not} partially dominated (in the sense described following Definition 9.3 in \cite{douc:2014}). Though a number of developments described in \cite{douc:2014} apply or extend easily to partially non-dominated models (as in the next remark), additional issues remain.
\end{Remark}

\begin{Remark}
\label{r:HMM-filtering}
When our model is an HMM with, for example, the underlying Gaussian AR$(1)$ series as in the preceding remark, the algorithm described in (\ref{e:X0-new})--(\ref{e:X-update-pred-again}) is the SIS algorithm discussed in Section 10.2 of \cite{douc:2014} with the choice of the optimal kernel and the associated weight function in Eqs.\ (10.30) and (10.31) of \cite{douc:2014}. This can be seen from the following observations. For an AR$(1)$ series, the one-step-ahead prediction is $\hat Z_{t+1} = \phi Z_{t}$ (and $\hat{z}_{t+1} = \phi z_{t}$). Though as noted in the preceding remark, our HMM model is not partially dominated and hence a transition density function $g(z,x)$ (defined following Definition 9.3 of \cite{douc:2014}) is not available, a number of formulas for partially dominated HMMs given in \cite{douc:2014} also apply to our model by taking
\begin{equation}
\label{e:ar1-g}
g(z,k) = 1_{A_k}(z).
\end{equation}
This is the case for the developments in Section 10.2 on SIS in \cite{douc:2014}. For example, one could check with (\ref{e:ar1-g}) that the filtering distribution of $\phi_t$ in Eq.\ (10.23) of \cite{douc:2014} is exactly that in (\ref{e:cond expectation}). The kernel $Q_t(z,A)$ appearing in Section 10.2 of \cite{douc:2014} is then
\begin{equation}
\label{e:kernel-Q}
Q_t(z,A) = \int_A M(z,dz') g(z',x_t) = \int_{A \cap A_{x_t}}
\frac{e^{-\frac{(z'-\phi z)^2}{2(1-\phi^2)}} dz'}{\sqrt{2 \pi(1-\phi^2)}},
\end{equation}
where (\ref{e:ar1-kernelX}) and (\ref{e:ar1-g}) were used. Sampling $Z_t^i$ from the kernel $Q_t(Z_{t-1}^i,\cdot)/Q_t(Z_{t-1}^i,\RR)$ (see p.\ 330 in \cite{douc:2014}) can be shown to be equivalent to defining $Z_t^i$ through Steps 1 and 2 of our algorithm in (\ref{e:X0-new})--(\ref{e:X-update-pred-again}). The optimal weight function $Q_t(z,\RR)$ can also be checked to be that in (\ref{e:explicit-truncNormal}) above.
\end{Remark}

\section{Additional simulations}
\label{a:simulations}

This section expands the simulation study of Section \ref{s:simulations}.   For the reader's convenience, we first list common count marginal distribution forms, some of which are used in our simulations.
\begin{itemize}
\item Binomial (Bin($N,p$)): $\PP [ X_t =k] =  {N \choose k} p^k (1-p)^{N-k}$, $k \in \{0, \ldots, N \}$, $p \in (0,1)$;
\item Poisson (Pois($\lambda$)): $\PP[ X_t =k] = e^{-\lambda} \lambda^k/k!$, with $\lambda>0$;
\item Mixture Poisson (mixPois($\boldsymbol{\lambda},\boldsymbol{p}$)): $\PP[ X_t =k] = \sum_{m=1}^M p_m e^{-\lambda_m} \lambda_m^k/k!$, where $\boldsymbol{p} = (p_1,\ldots,p_M)^\prime$ with the mixture probabilities $p_m>0$ such that $\sum_{m=1}^M p_m=1$ and $\boldsymbol{\lambda} = (\lambda_1,\ldots,\lambda_M)^\prime$ with $\lambda_m>0$ for each $m$;
\item Negative binomial (NB($r,p$)): $\PP[X_t=k]=\frac{\Gamma(r+k)}{k!\Gamma(r)}(1-p)^r p^k$, with $r > 0$ and $p \in (0,1)$;
\item Generalized Poisson (GPois($\lambda,\eta$)): $\PP[X_t =k] = e^{-(\lambda+\eta k)} \lambda(\lambda+\eta k)^{k-1}/k!$, with $\lambda > 0$ and $\eta \in [0,1)$;
\item Conway-Maxwell-Poisson (CMP$(\lambda,\nu)$): $\PP[X_t=k]= \frac{\lambda^k}{(k!)^\nu C(\lambda, \nu)}$, with $\lambda> 0$, $\nu > 0$, and a normalizing constant $C(\lambda,\nu)$ making the probabilities sum to unity.
\end{itemize}
The CDFs of the mixture Poisson and negative binomial distributions are denoted below by $\mathcal{MP}$ and $\mathcal{NB}$, respectively. As in Section  \ref{s:simulations}, estimates of a parameter $\zeta$ from Gaussian pseudo-likelihood (GL), implied Yule-Walker (IYW), and particle filtering (PF/SMC) methods are denoted by $\hat{\zeta}_{GL}$, $\hat{\zeta}_{IYW}$, and $\hat{\zeta}_{PF}$, respectively.

\subsection{Mixed Poisson AR(1)}
\label{s:simulations-mixedpoisson-ar1}
Consider the three-parameter mixture Poisson marginal distribution with parameters $\lambda_1 >0, \lambda_2>0$, and $p \in [0,1]$, and probability mass function as defined above.  As in Section \ref{s:simulations}, the count series was obtained by transforming the AR(1) process $Z_t = \phi Z_{t-1} + (1-\phi^2)^{1/2}\epsilon_t$ via \eqref{e:Y-GX} with $F=\mathcal{MP}$. Eight parameter schemes that consider all combinations of $\lambda_1=2$, $\lambda_2 \in \{ 5, 10 \}$, $p= 0.25$, and $\phi=\{ \pm 0.25, \pm 0.75 \}$ are studied.

Figure \ref{f:mixedpois-ar1-sim} shows box plots of the parameter estimates for $\phi=0.75$ and $\lambda_2=5$ or 10 (left or right panels, respectively). To ensure parameter identifiability, $p$ was constrained to lie in $(0, 1/2)$. In the $\lambda_2=5$ case (left panel), PF/SMC methods outperform GL and IYW approaches, yielding smaller biases and variances for most parameter choices and all sample sizes. The only exception occurs with $\lambda_2$, where $\hat{\lambda}_{2,GL}$ were moderately superior to $\hat{\lambda}_{2,PF}$ and $\hat{\lambda}_{2,IYW}$ for $T=100$ and  $T=200$; however for $T=400$, PF/SMC performs well, having little bias and the smallest variance of the three methods.   IYW produced significantly smaller biases than GL in estimating $\lambda_1$ and $p$, but both methods estimate $\phi$ somewhat biasedly.  IYW also displays larger variances for estimates of $\lambda_1, \lambda_2$, and $p$ when $T$ is smaller.

In the $\lambda_2=10$ case (right panel), where bimodality features are more pronounced, the GL method performs (as one might expect) quite poorly. Here, the probability that $X_t$ is close to its mean value of $p \lambda_1 + (1-p) \lambda_2$ is actually quite small, but GL overestimates it as the mode of the corresponding Gaussian distribution with that mean. In contrast, the PF/SMC approach ``feels the entire joint distribution of the process", outperforming the IYW and GL approaches across the board. IYW also does reasonably well in this setting, although not quite as good as PF/SMC.

\begin{figure}[h!]
\centerline{
\includegraphics[width = 0.5\textwidth]{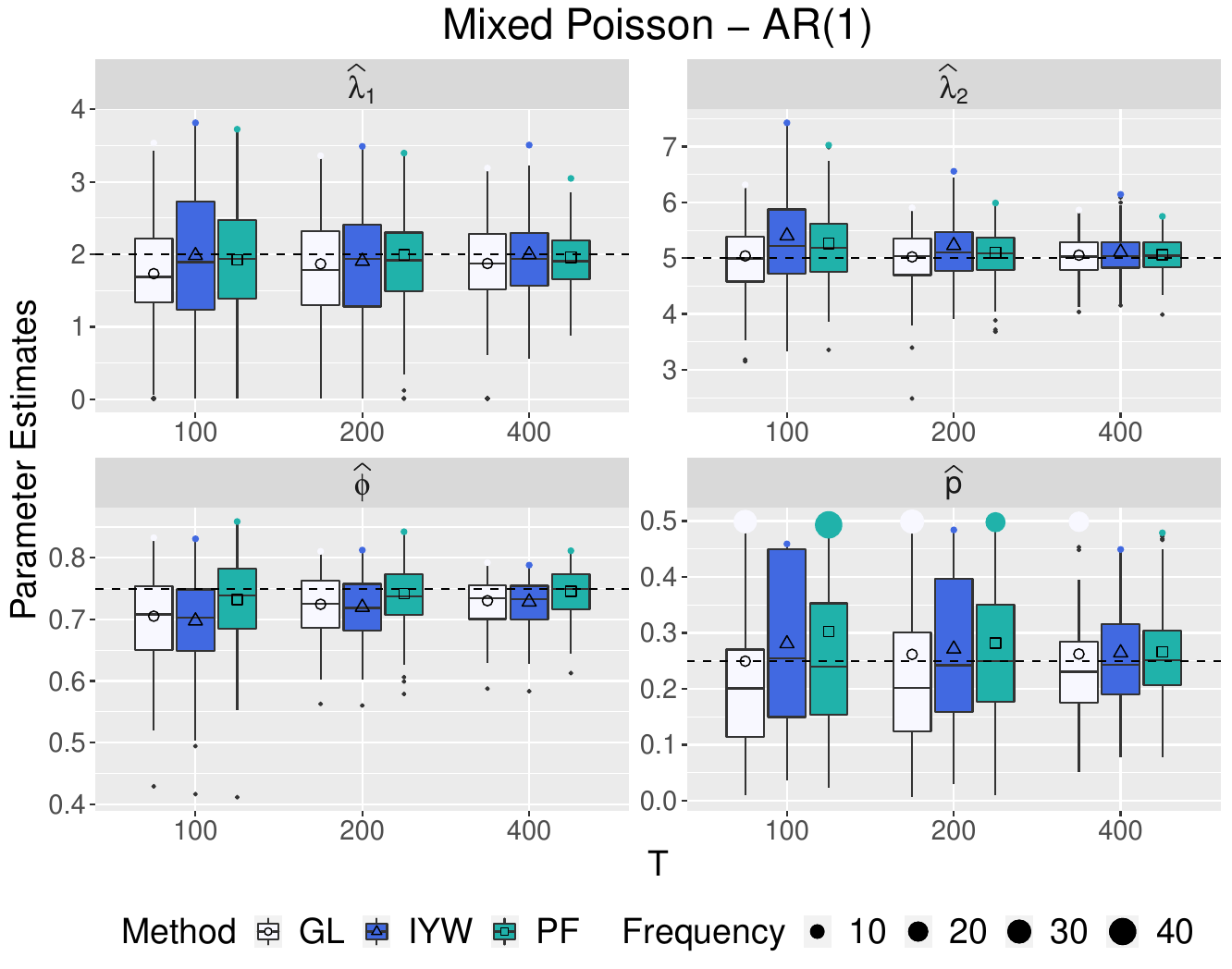}
\includegraphics[width = 0.5\textwidth]{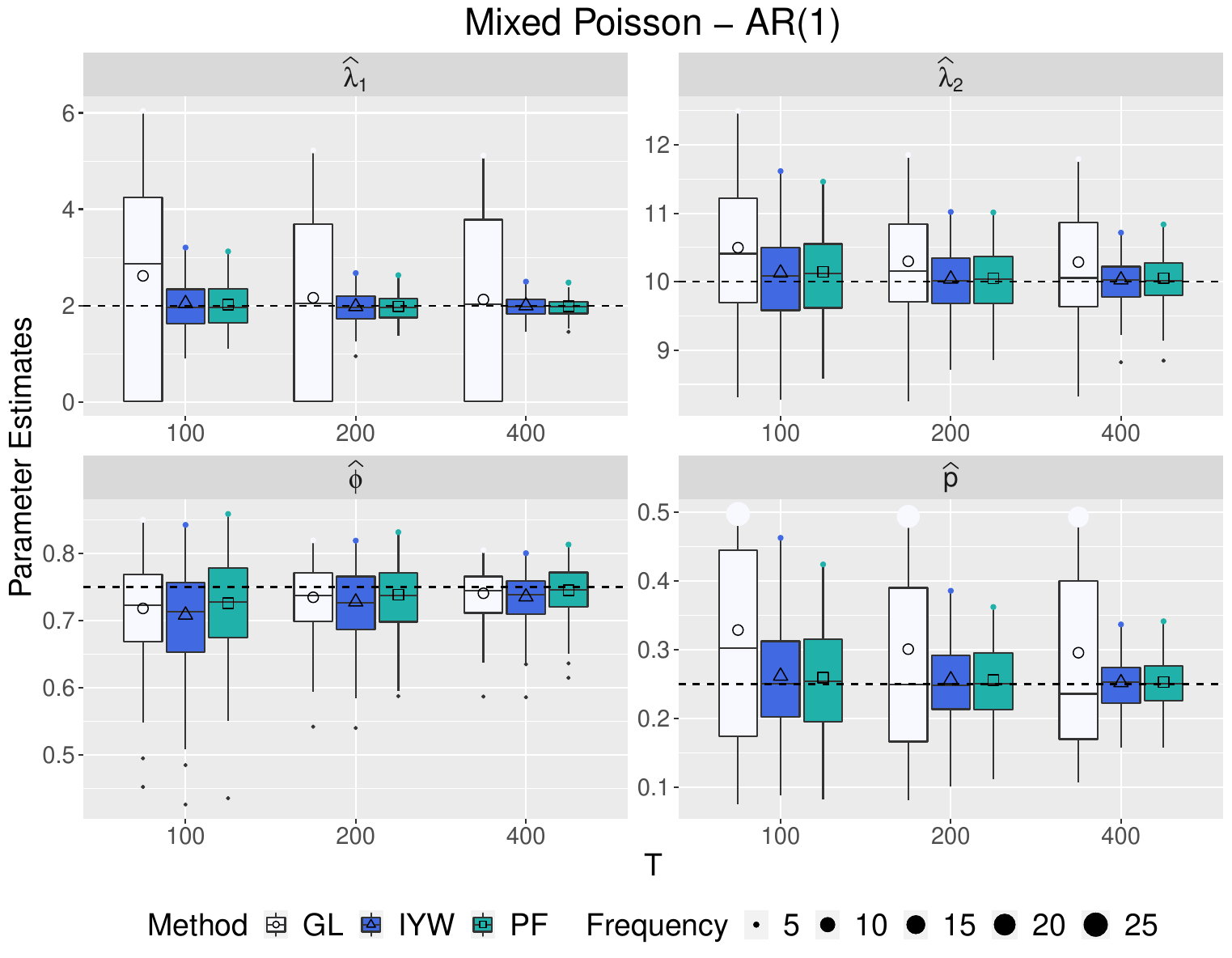}
}
\caption{\label{f:mixedpois-ar1-sim}
\textit{Gaussian likelihood, implied Yule-Walker, and PF/SMC parameter estimates for 200 synthetic mixed Poisson AR(1) series of lengths $T = 100, 200$, and $400$. The true parameter values (indicated by the black horizontal dashed lines) are $\lambda_1=2$, $\lambda_2=5$, $\phi = 0.75$, and $p=1/4$ (left panel) and $\lambda_1=2$, $\lambda_2=10$, $\phi = 0.75$, and $p=1/4$ (right panel).}}
\end{figure}

\subsection{Negative binomial MA(1)}
\label{s:simulations-negbinom-ma1}
Our final case considers the negative binomial distribution with parameters $r>0$, $p \in (0,1)$, and probability mass function as defined above.  To obtain $X_t$, the MA(1) process
\begin{equation}
Z_t = \epsilon_t + \theta \epsilon_{t-1},
\end{equation}
was simulated and transformed via \eqref{e:Y-GX} with $F=\mathcal{NB}$; $\EE[ Z_t^2] \equiv 1$ was induced by taking $\mbox{Var}(\epsilon_t) = (1+\theta^2)^{-1}$. Eight parameter schemes resulting from all combinations of $p \in \{0.2, 0.5\}$, $r=3$, and $\theta \in\{ \pm 0.25, \pm 0.75 \}$ were considered.   The negative binomial marginal distribution is overdispersed.   Since $\{ Z_t \}$ is not an autoregression, IYW estimates are not considered.

Figure \ref{f:negbinom-ma1-sim} displays box plots of parameter estimates from models with $\theta = 0.75$ (left panel) and $\theta = -0.75$ (right panel). The PF/SMC approach is clearly superior here for all parameters and sample sizes.  GL estimates incur ``boundary issues'' with $\hat{\theta}_{GL}$ for small $T$ and negatively correlated series (right panel).  Elaborating, we impose $\hat{\theta}$ to lie in $(-1,1)$ for an invertible moving-average and some GL runs ``press this estimate" out to $-1$.  GL boundary issues (and any biases) dissipate and sampling variability decreases appreciably with the largest series length $T=400$; this said, PF/SMC still performs best.

\begin{figure}[h]
\centerline{
\includegraphics[width = 0.5\textwidth]{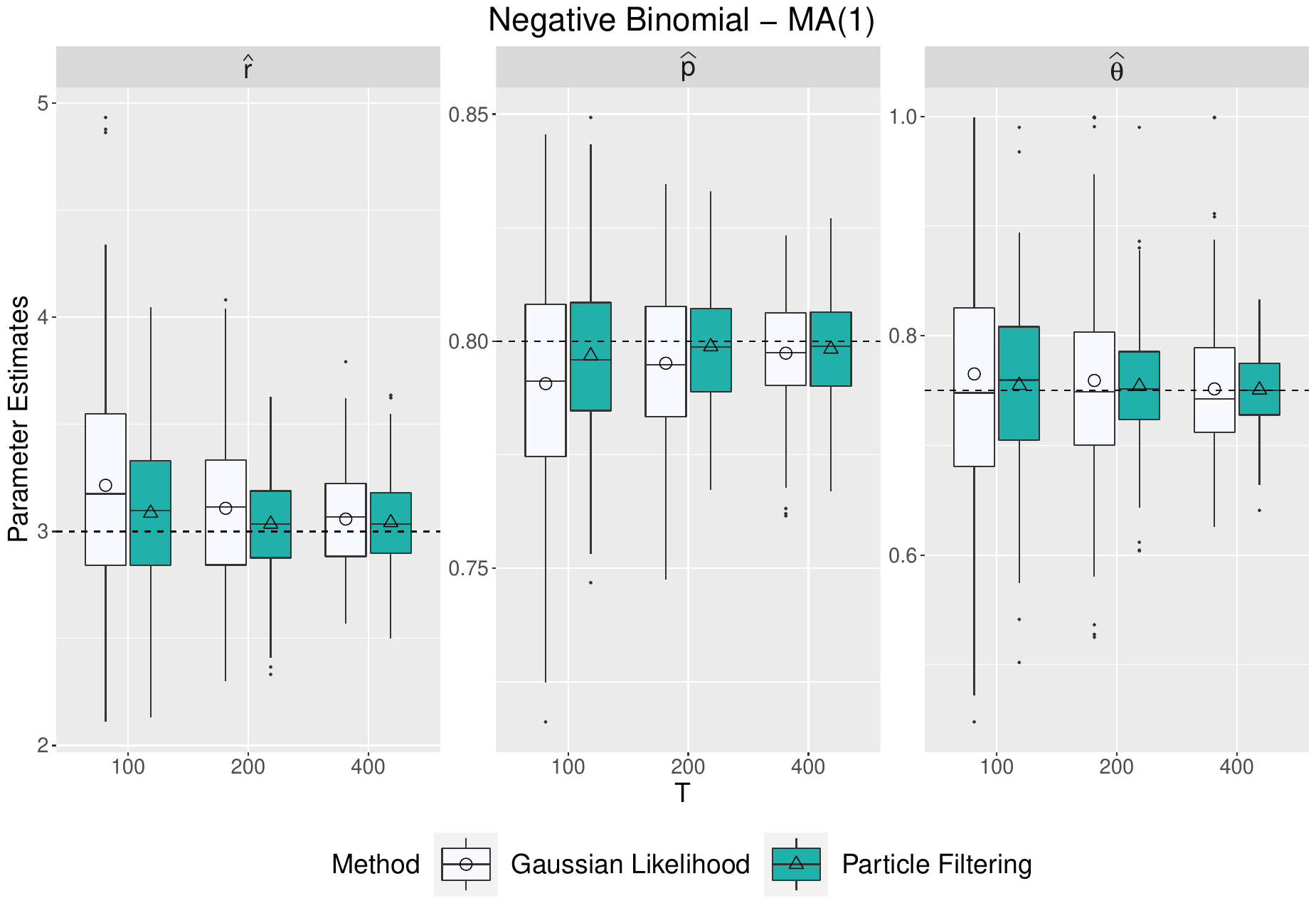}
\includegraphics[width = 0.5\textwidth]{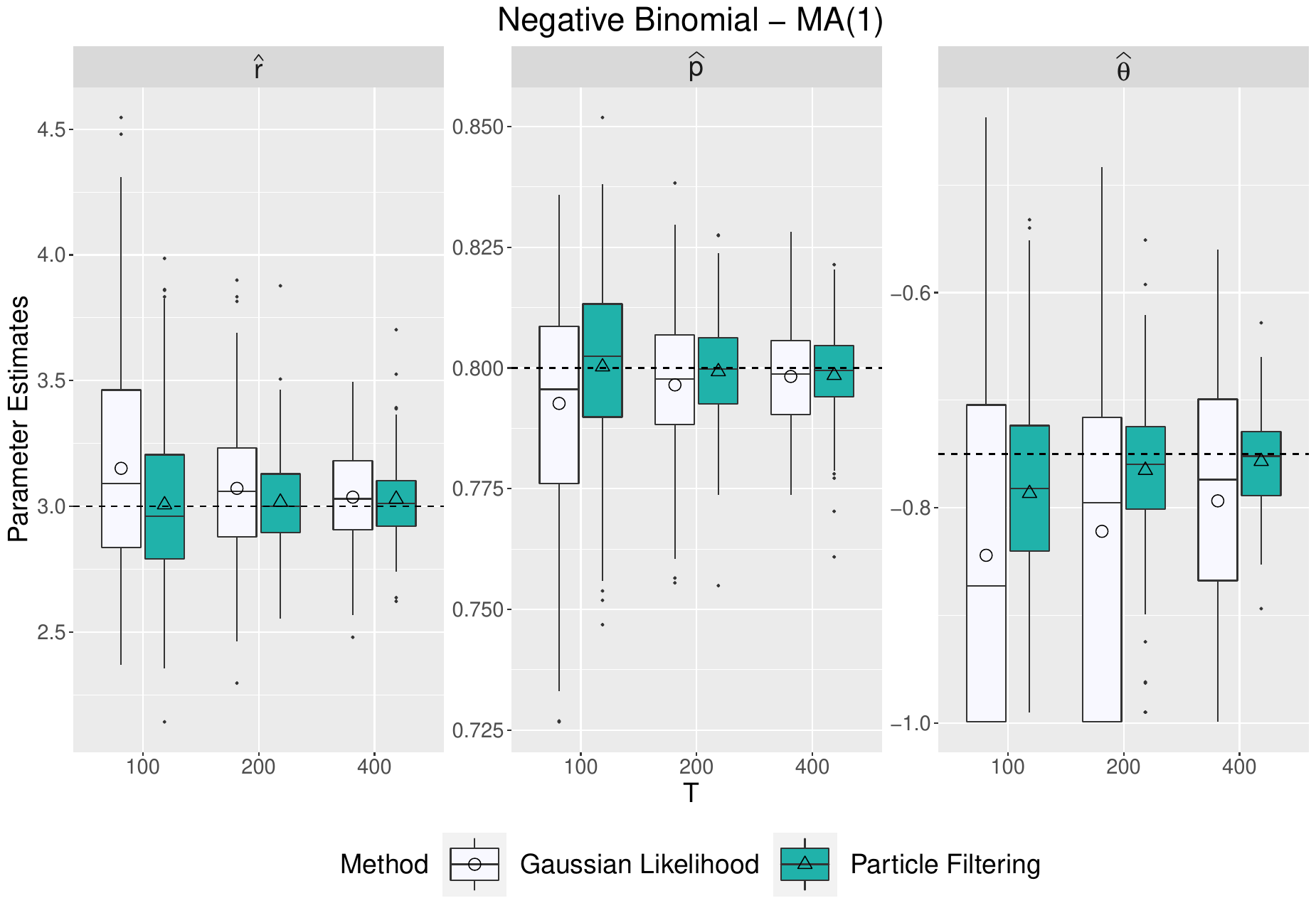}
}
\caption{\label{f:negbinom-ma1-sim} \textit{Gaussian likelihood and PF/SMC parameter estimates for 200 synthetic negative binomial MA(1)  series of lengths $T = 100, 200$ and $400$. The true parameter values (indicated by the horizontal dashed lines)
are $r=3$, $p=0.2$ and $\theta=0.75$ (left panel) and $r=3$, $p=0.2$ and $\theta=-0.75$ (right panel).}}
\end{figure}

Overall, PF/SMC likelihood methods exhibit the best performance, with the simple moment IYW methods being serviceable in the case where $\{ Z_t \}$ is an autoregression.  PF/SMC techniques were also recommended by \cite{HanDeOliveria_2018} (compared to other likelihood approximations) in spatial settings.

\begin{Remark}
We conducted an ANOVA-type experiment to numerically quantify the PF/SMC approximation error and compare its magnitude against estimation bias.   Specifically, our simulation study fitted each realization (from a total of 200) of a Poisson-AR(12) model five times (we do not list the chosen AR coefficients, but results are reasonably robust to their choice).  For each fit, the  particle numbers $N \in \{ 5, 10, 100, 500 \}$ are considered.  For each $N$, the 200 5-tuples of parameter estimates can be viewed as 200 ANOVA treatments, where the between- and within-treatments variations quantify the estimation and approximation error respectively. We found that the estimation error dominated the PF/SMC approximation error by several orders of magnitude, even with the smallest number of particles. While detailed results are omitted for brevity's sake, the inference is that the PF/SMC likelihood approximation is reasonably accurate in this setting.	
\end{Remark}

\section{Additional application tables}
\label{a:application-tables}

The following two tables compliment results presented in the applications section of the paper.

\begin{table}[h!]
	\centering
	{\footnotesize
	\begin{tabular}{|c|c|ccccccc|}
	\hline
	Marginal Distribution &Model                & WN & AR(1)    & AR(2)   & AR(3) & MA(1)    & MA(2) & MA(3)  \\
	\hline
	 \multirow{4}{*}{negative binomial}
	    & $\mbox{AICc}_{GL}$ & 844.2  & 827.5           & 828.4 & 829.9   & 834.2  & \textbf{825.4}  & 825.4 \\
	 {} & $\mbox{BIC}_{GL}$  & 851.9  & \textbf{837.7}  & 841.0   & 844.9   & 844.4  & 838             & 840.4 \\
	 \cline{2-9}
	 {} & $\mbox{AICc}_{PF}$ & 748.5  & 736.9  & 732.0   & \textbf{721.7}  & 741.5  & 730.3 & 729.9\\
	 {} & $\mbox{BIC}_{PF}$  & 756.2   & 747.1  & 744.6 & \textbf{736.7}  & 751.7  & 742.9 & 744.9 \\
	\hline
	 \multirow{4}{*}{generalized Poisson}
	    & $\mbox{AICc}_{GL}$ & 847.5 & 830.6 & 831.0  & 833.3 & 836.9 & 828.7 & \textbf{828.6} \\
	 {} & $\mbox{BIC}_{GL}$  & 855.2 & \textbf{840.8} & 843.6  & 848.3 & 847.0 & 841.3 & 843.6 \\
	 \cline{2-9}
	 {} & $\mbox{AICc}_{PF}$ & 769.2 & 754.1 & 749.8  & \textbf{741.2} & 758.6 & 749.8 & 749.9 \\
	 {} & $\mbox{BIC}_{PF}$  & 776.9 & 764.3 & 762.4  & \textbf{756.2} & 768.8 & 762.4 & 764.9 \\
	\hline
	\end{tabular}}
	\caption{AIC and BIC statistics for generalized Poisson and negative binomial distributions with different latent Gaussian ARMA orders.}
	\label{a:ARMAfits-2}
	\end{table}

	\begin{table}[h!]
		\centering
		{\small
		\begin{tabular}{|c|cccccc|}
		\hline
		Parameters           & $\phi_1$ & $\phi_2$ & $\phi_3$ & $\beta_0$ & $\beta_1$  & $a$\\
		\hline
		GL Estimates         &   -0.498 &   0.104 &   0.188  & 2.450 &  0.467 &   0.201\\
		GL Standard Errors   &    0.245 &  0.250  &  0.161   & 0.095 & 0.108  & 0.036   \\
		\hline
		PF/SMC Estimates         & -0.331 & 0.178  & 0.232  & 2.211  & 1.038 & 0.298  \\
		PF/SMC Standard Errors   &  0.086 & 0.098  & 0.091  & 0.144  & 0.301 & 0.038 \\
		\hline
		\end{tabular}}
		\caption{Estimates and standard errors of the generalized Poisson-AR(3) model.}
		\label{GenPoisAR3fit}
	\end{table}
\newpage

\bibliographystyle{Chicago}

\bibliography{bib_latentGaussCounts}

\begin{thebibliography}{}

\bibitem[\protect\citeauthoryear{Asmussen}{Asmussen}{2014}]{Asmussen_2014}
Asmussen, S. (2014).
\newblock Modeling and performance of bonus-malus systems: stationarity versus
  age-correction.
\newblock {\em Risks\/}~{\em 2}, 49--73.

\bibitem[\protect\citeauthoryear{Belyaev, Burnaev, and Kapushev}{Belyaev
  et~al.}{2015}]{Gammerman_etal_2015}
Belyaev, M., E.~Burnaev, and Y.~Kapushev (2015).
\newblock Gaussian process regression for structured data sets.
\newblock In A.~Gammerman, V.~Vovk, and H.~Papadopoulos (Eds.), {\em
  Statistical Learning and Data Sciences: Third International Symposium, SLDS
  2015}. Switzerland: Springer International Publishing.

\bibitem[\protect\citeauthoryear{Benjamin, Rigby, and Stasinopoulos}{Benjamin
  et~al.}{2003}]{GARMA}
Benjamin, M.~A., R.~A. Rigby, and D.~M. Stasinopoulos (2003).
\newblock Generalized autoregressive moving average models.
\newblock {\em Journal of the American Statistical Association\/}~{\em 98},
  214--223.

\bibitem[\protect\citeauthoryear{Berahas, Byrd, and Nocedal}{Berahas
  et~al.}{2019}]{berahas:2019}
Berahas, A.~S., R.~H. Byrd, and J.~Nocedal (2019).
\newblock Derivative-free optimization of noisy functions via quasi-{N}ewton
  methods.
\newblock {\em SIAM Journal on Optimization\/}~{\em 29}, 965--993.

\bibitem[\protect\citeauthoryear{Blight}{Blight}{1989}]{Blight_1989}
Blight, P.~A. (1989).
\newblock Time series formed from the superposition of discrete renewal
  processes.
\newblock {\em Journal of Applied Probability\/}~{\em 26}, 189--195.

\bibitem[\protect\citeauthoryear{Brockwell and Davis}{Brockwell and
  Davis}{1991}]{brockwell:davis:1991}
Brockwell, P.~J. and R.~A. Davis (1991).
\newblock {\em Time Series: Theory and Methods\/} (Second ed.).
\newblock New York City: Springer-Verlag.

\bibitem[\protect\citeauthoryear{Cario and Nelson}{Cario and
  Nelson}{1997}]{Cairo_Nelson_1997}
Cario, M.~C. and B.~L. Nelson (1997).
\newblock Modeling and generating random vectors with arbitrary marginal
  distributions and correlation matrix.
\newblock {\em Technical Report\/}~{\em Department of Industrial Engineering
  and Management Sciences, Northwestern University}.

\bibitem[\protect\citeauthoryear{Chen}{Chen}{2001}]{HChen_2001}
Chen, H. (2001).
\newblock Initialization of {NORTA}: Generation of random vectors with
  specified marginals and correlations.
\newblock {\em Inform Journal on Computing\/}~{\em 13}, 312--331.

\bibitem[\protect\citeauthoryear{Chopin and Papaspiliopoulos}{Chopin and
  Papaspiliopoulos}{2020}]{chopin:2020}
Chopin, N. and O.~Papaspiliopoulos (2020).
\newblock {\em An Introduction to Sequential Monte Carlo Methods}.
\newblock New York City: Springer.

\bibitem[\protect\citeauthoryear{Cui and Lund}{Cui and
  Lund}{2009}]{CuiLund_2009}
Cui, Y. and R.~B. Lund (2009).
\newblock A new look at time series of counts.
\newblock {\em Biometrika\/}~{\em 96}, 781--792.

\bibitem[\protect\citeauthoryear{Czado, Gneiting, and Held}{Czado
  et~al.}{2009}]{czado:etal:2009}
Czado, C., T.~Gneiting, and L.~Held (2009).
\newblock Predictive model assessment for count data.
\newblock {\em Biometrics\/}~{\em 65}, 1254--1261.

\bibitem[\protect\citeauthoryear{Davis, Holan, Lund, and Ravishanker}{Davis
  et~al.}{2016}]{Davis_etal_2015}
Davis, R.~A., S.~H. Holan, R.~B. Lund, and N.~Ravishanker (Eds.) (2016).
\newblock {\em Handbook of Discrete-Valued Time Series}.
\newblock Boca Raton, Florida, USA: CRC Press.

\bibitem[\protect\citeauthoryear{De~Oliveira}{De~Oliveira}{2016}]{DeOliveria_2013}
De~Oliveira, V. (2016).
\newblock Hierarchical {P}oisson models for spatial count data.
\newblock {\em Journal of Multivariate Analysis\/}~{\em 122}, 393--408.

\bibitem[\protect\citeauthoryear{Douc, Moulines, and Stoffer}{Douc
  et~al.}{2014}]{douc:2014}
Douc, R., E.~Moulines, and D.~S. Stoffer (2014).
\newblock {\em Nonlinear Time Series: Theory, Methods, and Applications with R
  Examples}.
\newblock Boca Raton, Florida, USA: CRC Press.

\bibitem[\protect\citeauthoryear{Doucet, De~Freitas, and Gordon}{Doucet
  et~al.}{2001}]{doucet:2001}
Doucet, A., N.~De~Freitas, and N.~Gordon (2001).
\newblock An introduction to sequential {M}onte {C}arlo methods.
\newblock In {\em Sequential {M}onte {C}arlo Methods in Practice}, pp.\  3--14.
  New York City: Springer.

\bibitem[\protect\citeauthoryear{Doucet and Johansen}{Doucet and
  Johansen}{2009}]{doucet:2009}
Doucet, A. and A.~M. Johansen (2009).
\newblock A tutorial on particle filtering and smoothing: Fifteen years later.
\newblock {\em Handbook of Nonlinear Filtering\/}~{\em 12}, 656--704.

\bibitem[\protect\citeauthoryear{Dumsmuir}{Dumsmuir}{2016}]{Dunsmuir}
Dumsmuir, W.~T.~M. (2016).
\newblock Generalized linear autoregressivee moving-averge models.
\newblock In {\em Handbook of Discrete-valued Time Series}, pp.\  51--76. Boca
  Raton, Florida, USA: CRC Press.

\bibitem[\protect\citeauthoryear{Famoye}{Famoye}{1993}]{famoye1993}
Famoye, F. (1993).
\newblock Restricted generalized {P}oisson regression model.
\newblock {\em Communications in Statistics-Theory and Methods\/}~{\em 22},
  1335--1354.

\bibitem[\protect\citeauthoryear{Fokianos}{Fokianos}{2012}]{fokianos:2012}
Fokianos, K. (2012).
\newblock Count time series models.
\newblock In {\em Handbook of Statistics}, Volume~30, pp.\  315--347.
  Amsterdam: Elsevier.

\bibitem[\protect\citeauthoryear{Freedman}{Freedman}{2006}]{Freedman_2006}
Freedman, D. (2006).
\newblock On the so-called ``{H}uber sandwich estimator" and ``robust standard
  errors".
\newblock {\em The American Statistician\/}~{\em 60}, 299--302.

\bibitem[\protect\citeauthoryear{Genz}{Genz}{2004}]{genz:2004numerical}
Genz, A. (2004).
\newblock Numerical computation of rectangular bivariate and trivariate normal
  and t probabilities.
\newblock {\em Statistics and Computing\/}~{\em 14\/}(3), 251--260.

\bibitem[\protect\citeauthoryear{Grigoriu}{Grigoriu}{2007}]{Grigoriu_2007}
Grigoriu, M. (2007).
\newblock Multivariate distributions with specified marginals: applications to
  wind engineering.
\newblock {\em Journal of Engineering Mechanics\/}~{\em 133}, 174--184.

\bibitem[\protect\citeauthoryear{Hajivassiliou and Ruud}{Hajivassiliou and
  Ruud}{1994}]{hajivassiliou:ruud:1994}
Hajivassiliou, V.~A. and P.~A. Ruud (1994).
\newblock Classical estimation methods for {LDV} models using simulation.
\newblock In {\em Handbook of Econometrics, {V}ol. {IV}}, Volume~2 of {\em
  Handbooks in Econometrics}, pp.\  2383--2441. Amsterdam: North-Holland.

\bibitem[\protect\citeauthoryear{Han and De~Oliveira}{Han and
  De~Oliveira}{2016}]{HanDeOliveria_2016}
Han, Z. and V.~De~Oliveira (2016).
\newblock On the correlation structure of {G}aussian copula models for
  geostatistical count data.
\newblock {\em Australian \& New Zealand Journal of Statistics\/}~{\em 58},
  47--69.

\bibitem[\protect\citeauthoryear{Han and De~Oliveira}{Han and
  De~Oliveira}{2020}]{HanDeOliveria_2018}
Han, Z. and V.~De~Oliveira (2020).
\newblock Maximum likelihood estimation of {G}aussian copula models for
  geostatistical count data.
\newblock {\em Communications in Statistics - Simulation and
  Computation\/}~{\em 49}, 1957--1981.

\bibitem[\protect\citeauthoryear{Jacobs and Lewis}{Jacobs and
  Lewis}{978a}]{JacobsLewis_1978a}
Jacobs, P.~A. and P.~A.~W. Lewis (1978a).
\newblock Discrete time series generated by mixtures {I}: Correlational and
  runs properties.
\newblock {\em Journal of the Royal Statistical Society\/}~{\em 40}, 94--105.

\bibitem[\protect\citeauthoryear{Joe}{Joe}{1996}]{Joe_1996}
Joe, H. (1996).
\newblock Time series models with univariate margins in the convolution-closed
  infinitely divisible class.
\newblock {\em Journal of Applied Probability\/}~{\em 33}, 664--677.

\bibitem[\protect\citeauthoryear{Joe and Zhu}{Joe and Zhu}{2005}]{joe:2005}
Joe, H. and R.~Zhu (2005).
\newblock Generalized {P}oisson distribution: the property of mixture of
  {P}oisson and comparison with negative binomial distribution.
\newblock {\em Biometrical Journal: Journal of Mathematical Methods in
  Biosciences\/}~{\em 47}, 219--229.

\bibitem[\protect\citeauthoryear{Kachour and Yao}{Kachour and
  Yao}{2009}]{Kachour_Yao_2009}
Kachour, M. and J.~F. Yao (2009).
\newblock First order rounded integer valued autoregressive {(RINAR(1))}
  processes.
\newblock {\em Journal of Time Series Analysis\/}~{\em 30}, 417--448.

\bibitem[\protect\citeauthoryear{Kantas, Doucet, Singh, Maciejowski, and
  Chopin}{Kantas et~al.}{2015}]{kantas:2015}
Kantas, N., A.~Doucet, S.~S. Singh, J.~Maciejowski, and N.~Chopin (2015).
\newblock On particle methods for parameter estimation in state-space models.
\newblock {\em Statistical Science\/}~{\em 30}, 328--351.

\bibitem[\protect\citeauthoryear{Kedem}{Kedem}{1980}]{Kedem_1980}
Kedem, B. (1980).
\newblock Estimation of the parameters in stationary autoregressive processes
  after hard limiting.
\newblock {\em Journal of the American Statistical Association\/}~{\em 75},
  146--153.

\bibitem[\protect\citeauthoryear{Kolassa}{Kolassa}{2016}]{kolassa:2016}
Kolassa, S. (2016).
\newblock Evaluating predictive count data distributions in retail sales
  forecasting.
\newblock {\em International Journal of Forecasting\/}~{\em 32}, 788--803.

\bibitem[\protect\citeauthoryear{Lennon}{Lennon}{2016}]{Lennon_2016}
Lennon, H. (2016).
\newblock {\em Gaussian copula modelling for integer-valued time series}.
\newblock Ph.\ D. thesis, The University of Manchester.

\bibitem[\protect\citeauthoryear{Liu}{Liu}{2008}]{liu:2008}
Liu, J.~S. (2008).
\newblock {\em Monte Carlo Strategies in Scientific Computing}.
\newblock New York City: Springer Science \& Business Media.

\bibitem[\protect\citeauthoryear{Liu and Chen}{Liu and Chen}{1998}]{liu:1998}
Liu, J.~S. and R.~Chen (1998).
\newblock Sequential {M}onte {C}arlo methods for dynamic systems.
\newblock {\em Journal of the American Statistical Association\/}~{\em 93},
  1032--1044.

\bibitem[\protect\citeauthoryear{Livsey, Lund, Kechagias, and Pipiras}{Livsey
  et~al.}{2018}]{Livsey_etal_2018}
Livsey, J., R.~B. Lund, S.~Kechagias, and V.~Pipiras (2018).
\newblock Multivariate integer-valued time series with flexible autocovariances
  and their application to major hurricane counts.
\newblock {\em Annals of Applied Statistics\/}~{\em 12}, 408--431.

\bibitem[\protect\citeauthoryear{Malik and Pitt}{Malik and
  Pitt}{2011}]{malik:2011}
Malik, S. and M.~K. Pitt (2011).
\newblock Particle filters for continuous likelihood evaluation and
  maximisation.
\newblock {\em Journal of Econometrics\/}~{\em 165}, 190--209.

\bibitem[\protect\citeauthoryear{Masarotto and Varin}{Masarotto and
  Varin}{2012}]{masarotto:2012}
Masarotto, G. and C.~Varin (2012).
\newblock {G}aussian copula marginal regression.
\newblock {\em Electronic Journal of Statistics\/}~{\em 6}, 1517--1549.

\bibitem[\protect\citeauthoryear{Nash and Varadhan}{Nash and
  Varadhan}{2011}]{nash:2011}
Nash, J.~C. and R.~Varadhan (2011).
\newblock Unifying optimization algorithms to aid software system users: optimx
  for {R}.
\newblock {\em Journal of Statistical Software\/}~{\em 43}, 1--14.

\bibitem[\protect\citeauthoryear{Pipiras and Taqqu}{Pipiras and
  Taqqu}{2017}]{pipiras:2017}
Pipiras, V. and M.~S. Taqqu (2017).
\newblock {\em Long-Range Dependence and Self-Similarity}, Volume~45.
\newblock Boca Raton, Florida, USA: Cambridge University Press.

\bibitem[\protect\citeauthoryear{Shi, Xuan, Oztoprak, and Nocedal}{Shi
  et~al.}{2021}]{shi:2021}
Shi, H.-J.~M., M.~Q. Xuan, F.~Oztoprak, and J.~Nocedal (2021).
\newblock On the numerical performance of derivative-free optimization methods
  based on finite-difference approximations.
\newblock {\em arXiv preprint arXiv:2102.09762\/}.

\bibitem[\protect\citeauthoryear{Smith and Khaled}{Smith and
  Khaled}{2012}]{Smith2012}
Smith, M.~S. and M.~A. Khaled (2012).
\newblock Estimation of copula models with discrete margins via {B}ayesian data
  augmentation.
\newblock {\em Journal of the American Statistical Association\/}~{\em 107},
  290--303.

\bibitem[\protect\citeauthoryear{Song, Li, and Zhang}{Song
  et~al.}{2013}]{song2013}
Song, P., M.~Li, and P.~Zhang (2013).
\newblock Vector generalized linear models: a {G}aussian copula approach.
\newblock In P.~Jaworski, F.~Durante, and W.~H\"{a}rdle (Eds.), {\em Copulae in
  Mathematical and Quantitavie Finance}. Heidelberg, Germany: Springer.

\bibitem[\protect\citeauthoryear{Tong}{Tong}{1990}]{Tong}
Tong, Y.~L. (1990).
\newblock {\em The Multivariate Normal Distribution}.
\newblock New York City: Springer-Verlag.

\bibitem[\protect\citeauthoryear{Whitt}{Whitt}{1976}]{whitt:1976bivariate}
Whitt, W. (1976).
\newblock Bivariate distributions with given marginals.
\newblock {\em The Annals of Statistics\/}~{\em 4}, 1280--1289.

\bibitem[\protect\citeauthoryear{Zheng, Xiao, and Chen}{Zheng
  et~al.}{2015}]{zheng:2015}
Zheng, T., H.~Xiao, and R.~Chen (2015).
\newblock Generalized {ARMA} models with martingale difference errors.
\newblock {\em Journal of Econometrics\/}~{\em 189}, 492--506.

\end{thebibliography}
\end{document}